\documentclass[10pt]{article}

\usepackage{natbib}
\usepackage{acronym}
\usepackage{amsmath, amsfonts,amsthm,amssymb}
\usepackage{enumerate}
\usepackage[inline]{enumitem}   %
\usepackage{graphicx}
\usepackage{fullpage}
\usepackage[hidelinks]{hyperref}
 \usepackage{xcolor}

\acrodef{LC}[LC]{Lorenz curve}
\acrodef{LCs}[LCs]{Lorenz curves}
\acrodef{LRC}[LRC]{Lorenz random curve}
\acrodef{cdf}[cdf]{cumulative distribution function}
\acrodef{fPCA}[FPCA]{Functional Principal Component Analysis}
\acrodef{FPCA}[FPCA]{Functional Principal Component Analysis}

\newtheorem{theorem}{Theorem}%
\newtheorem{proposition}[theorem]{Proposition}%
\newtheorem{lemma}[theorem]{Lemma}%

\newtheorem{remark}{Remark}%

\newtheorem{definition}{Definition}%

\newcounter{HPcont}
\renewcommand{\theHPcont}{\mbox{(A-\arabic{HPcont})}}
\newenvironment{assume}{
	\begin{itemize}\refstepcounter{HPcont} \item[\textbf{\theHPcont}]} {\end{itemize}\par}

\begin{document}

\title%
{Principal Component Analysis for Zero-Inflated Probability Measures via Lorenz Curves}

\author{Giacomo Aletti\footnote{Dipartimento di Scienze e Politiche Ambientali, Universit\`{a} degli studi di Milano, giacomo.aletti@unimi.it},
	Enea G. Bongiorno\footnote{Dipartimento di Studi per l'Economia e l'Impresa, Universit\`{a}  del Piemonte Orientale, enea.bongiorno@uniupo.it} 
}

\maketitle

\begin{abstract}
	This paper develops a principal component analysis framework for probability measures associated with non-negative real-valued random variables with structural zeros.
	Specifically, leveraging the associated \ac{LC} and a suitable perspective transformation of the plane, a bijective embedding is defined to treat such measures as functional data in a Hilbert space. Once the random \ac{LC} is embedded in this space, standard \ac{FPCA} becomes directly applicable, and the diffeomorphic nature of the embedding allows the main convergence results and their interpretations to be transferred back to the original space of \acp{LC}. The a.s.\@ and $L^2(\mathbb{P})$ consistency of the estimated eigenfunctions, eigenvalues, and modes of variation is subsequently investigated, both in the embedded space and, by virtue of the embedding's diffeomorphism property, in the original \ac{LC} space. Crucially, this consistency analysis accounts for the double stochasticity of the estimation problem: one level of randomness is tied to the underlying \ac{LC} process itself, while the other stems from the empirical estimation of the \acp{LC} from sample data. As a consequence, consistency of the Karhunen--Loève reconstructions is also established, guaranteeing faithful approximation of individual curves in both spaces. Finally, the proposed \ac{FPCA} methodology is illustrated on a real-world dataset of global export flows and compared with a naive \ac{FPCA} approach.
\end{abstract}

\paragraph{keywords.}
Lorenz curves, Hilbert embedding approach, functional data analysis, perspective transformation, modes of variation

\section{Introduction}

Let $\mathcal{M}$ denote the family of probability measures on $[0,+\infty)$ that can be expressed as a convex combination of a Dirac measure at zero and a measure absolutely continuous with respect to the Lebesgue one. Specifically, any $\Gamma \in \mathcal{M}$ admits the representation
\begin{equation}\label{eq:misura_zero_inflated}
	\Gamma = \pi \delta_0 + (1-\pi)\lambda, \qquad \pi \in (0,1),
\end{equation}
where $\delta_0$ denotes the Dirac measure at zero and $\lambda$ is a probability measure absolutely continuous with respect to the Lebesgue one on $[0,+\infty)$. The presence of a positive probability mass at zero gives rise to what are commonly referred to as ``structural zeros'', and random variables associated to such measures are therefore known as ``zero-inflated''.

Measures of this form arise naturally in a wide range of applied domains including, for instance, financial markets, where variables such as trading volumes and number of trades exhibit a non-negligible proportion of exact zeros \citep{hau:mal:sch14}; health services, where medical expenditures are often zero for a large fraction of the population \citep{neelon2016modeling}; hydrology, where daily streamflows on intermittent rivers alternate between zero and positive values \citep{hutton1990non}; meteorology, where precipitation records frequently contain exact-zero amounts \citep{cuello2019predicting}; and inventory management, where intermittent and lumpy demand patterns produce large fractions of zero observations \citep{petropoulos2016another}.
The statistical analysis of zero-inflated data has traditionally been carried out through parametric models (see e.g.\@ \citealp{lambert1992} where zero-inflated Poisson regression is considered). These models usually impose strong distributional assumptions and, by focusing on a finite-dimensional summary, do not naturally extend to the settings in which the object of interest is the measure $\Gamma$ itself.
A related but distinct problem is that of PCA for zero-inflated compositional data, where zeros arise as numerical artefacts rather than structural probability mass; recent contributions in this direction include \cite{kim:par:jun24}, who propose a PCA method working directly on the compositional simplex without resorting to zero-replacement strategies. That setting, however, differs fundamentally from the present one, as the object of interest is a composition rather than a measure on $[0,+\infty)$.

A more flexible, geometry-aware approach consists of embedding the space of measures into a suitable Hilbert space, so that the rich toolkit of Functional Data Analysis (FDA -- see for instance the monographs \citealp{fer:vie06,hor:kok12,hsi:eub15,kok:rei17,ram:sil05}) becomes available. 
Several embeddings of this kind have been proposed in the literature. The Bayes space framework \citep{ego:dia:paw06, van2014bayes} provides a Hilbert space structure via the Radon--Nikodym derivative and the centered log-ratio transformation, and has been applied to \ac{FPCA} of density functions \citep{hro:men:tem:hru:fil16} and recently extended to mixture measures by \cite{mai:sto:fit:gre25}. \cite{pet:mul16} propose an alternative embedding via continuous and invertible transformations of the density, such as the log quantile density and the log hazard transformations. The work most directly related to the present paper is \cite{bon:goi19}, who construct a bijective embedding based on the \ac{LC} -- a function $L:[0,1]\to[0,1]$ summarizing the distribution of a non-negative random variable by describing the cumulative share of total wealth held by the bottom $p$ fraction of the population, and invariant to scale transformations -- by mapping its second derivative via the negative centered log-ratio transformation, thereby making FDA techniques directly applicable.
However, none of these frameworks handles zero-inflated measures $\Gamma \in \mathcal{M}$ from the \ac{LC} perspective: the \ac{LC} of such measures is flat on an initial interval $[0,\pi]$, with $\pi$ reflecting the structural zero mass, and strictly increasing thereafter, a feature that is not accommodated by \cite{bon:goi19}, which is restricted to strictly positive data.

The goal of this paper is to fill this gap by developing an \ac{FPCA} for samples of i.i.d.\@ measures with a.s.\@ finite mean $\Gamma_1, \ldots, \Gamma_n \in \mathcal{M}$, inspired by the framework of \cite{bon:goi19}. The proposed approach is based on a bijective embedding of $\mathcal{M}$ into $\mathbb{R} \times C([0,1])$, constructed by combining the \ac{LC} representation of $\Gamma$ with a perspective transformation $\varphi$ of the plane (see Figure~\ref{fig:intro_persp}). The perspective transformation acts on the graph of $L$ and produces a concave, twice differentiable curve $\ell$ defined on $[0,1]$: the flat segment of $L$ on $[0,\pi]$ collapses under $\varphi$ into a single intercept value $\ell(0) = \pi/(1-\pi)$, while the remaining shape of $L$ over $(\pi,1]$ is encoded in the concavity profile of $\ell$ over $(0,1]$. The map $L \mapsto (\ln(\ell(0)), \ln(-\ell''))$ then takes values in $\mathbb{R} \times C([0,1])$: the first component captures the log-odds of the structural zero mass, while the second encodes the inequality structure of the distribution, with $\ell''$ interpreted pointwise as a local radial weight of inequality (see Section~\ref{sec:the-psi}). It is shown that this embedding is a diffeomorphism from the family of \ac{LCs} $\mathcal{L}or$ onto $\mathbb{R} \times C([0,1])$, which is dense in the separable Hilbert space $\mathbb{R} \times L^2_{[0,1]}$; this allows FDA techniques to be applied in the embedded space, while the diffeomorphism property ensures that all convergence results transfer back to $\mathcal{L}or$, preserving theoretical guarantees and enhancing interpretability. The covariance operator of the embedded process has a block structure reflecting the decomposition into scalar and functional components, and its eigenfunctions are shown to be continuous, so that all finite-dimensional reconstructions via the Karhunen--Loève expansion are elements of $\mathbb{R} \times C([0,1])$ rather than merely of $\mathbb{R} \times L^2_{[0,1]}$.

\begin{figure}[h]
	\centering
	\includegraphics[width=.43\textwidth]{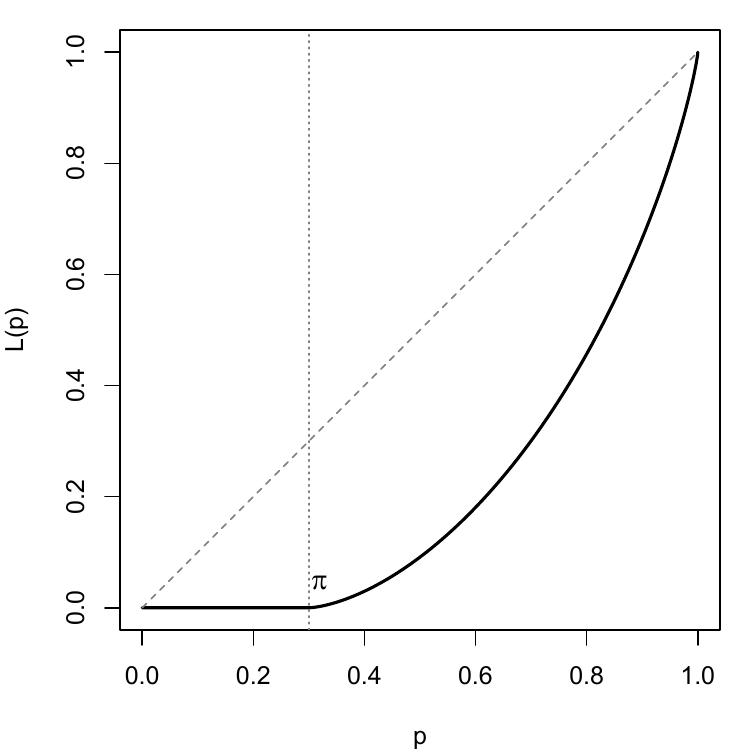}
	\includegraphics[width=.45\textwidth]{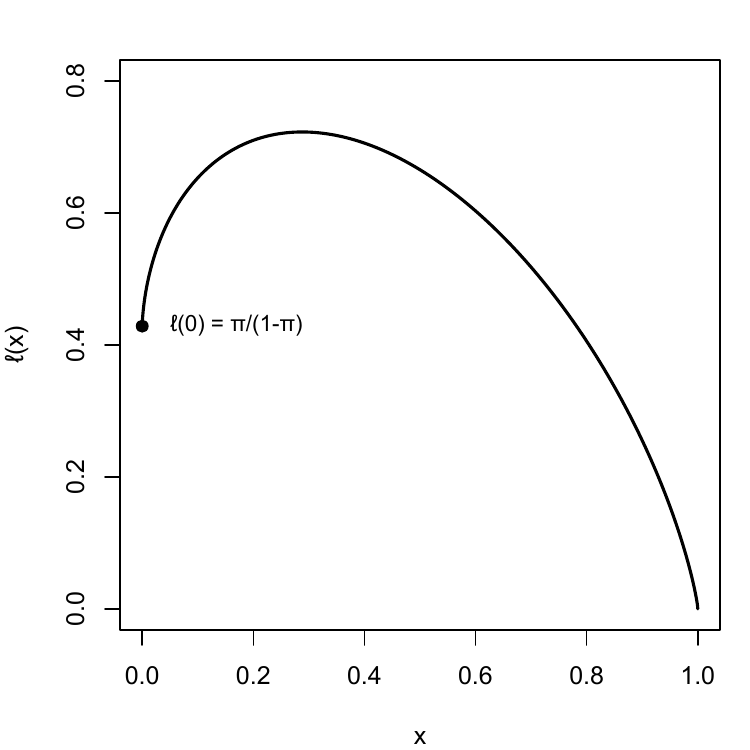}
	\caption{A simulated zero-inflated \ac{LC} $L \in \mathcal{L}or$ with 
		$\pi = 0.3$ (left) and its perspective transformation $\ell = 
		\varphi(L)$ (right). The flat segment $L(p)=0$ for $p \leq \pi$ is 
		mapped to the intercept $\ell(0) = \pi/(1-\pi)$, encoding the 
		structural zero mass as a single scalar.}
	\label{fig:intro_persp}
\end{figure}

It is worth contrasting the present approach with the Bayes space framework, which in principle could accommodate zero-inflated measures given a suitable choice of the reference measure. The two approaches differ, however, in the equivalence relation they impose and in the geometric structure they induce: the Bayes space framework embeds equivalence classes $[\Gamma] = \{\alpha\Gamma : \alpha > 0\}$ and lacks a direct connection to inequality measurement, whereas the present work embeds classes $[\mathbb{P}_X] = \{\mathbb{P}_{cX} : c > 0\}$ of distributions sharing the same relative inequality structure but differing in scale, placing the \ac{LC} interpretation at its core. Together with \cite{bon:goi19}, the two frameworks partition the space of non-negative distributions: $\pi = 0$ corresponds to the absolutely continuous case handled by \cite{bon:goi19}, while $\pi \in (0,1)$ is the zero-inflated case treated here.

The theoretical framework also accounts for the double stochasticity of the problem: a first source of randomness is related to the sampling of measures $\Gamma_1, \ldots, \Gamma_n$, and a second one to the fact that each $\Gamma_i$ is estimated from a finite sample of observations. Consistency of the eigenfunctions, eigenvalues, modes of variation, and Karhunen--Loève reconstructions is established both in $L^2(\mathbb{P})$ and almost surely -- a stronger result than those of \cite{bon:goi19}, who only establish $L^2(\mathbb{P})$ convergence -- and holds both in the embedded space and, by virtue of the diffeomorphism property, in the original \ac{LC} space.

The remainder of the paper is organized as follows. Section~\ref{sec:preliminaries} introduces the class $\mathcal{L}or$ of \ac{LCs}. Section~\ref{sec:diffeomorfismo} presents the perspective embedding and the bijection $\psi$, whereas Section~\ref{sec:consistency_results} illustrates the consistency results for the eigenfunctions, eigenvalues, modes of variation, and Karhunen--Loève reconstructions. Section~\ref{sec:application} illustrates the methodology on a real dataset concerning global export flows.
Appendix~\ref{sec:about_varphi} provides a geometric support to the introduced bijection while all proofs are collected in Appendix~\ref{sec:proofs}.

\section{Preliminaries on \ac{LCs}}
\label{sec:preliminaries}

Let $X$ be a non-negative real random variable with probability law $\Gamma \in \mathcal{M}$, with \ac{cdf} $F$ and finite mean $\mu$. Since $\Gamma$ admits the representation \eqref{eq:misura_zero_inflated}, the quantile function of $X$ satisfies $Q(p) = 0$ for $p \le \pi$ and is strictly positive and non-decreasing for $p > \pi$. The \ac{LC} associated to $X$ is defined according to \cite{gas72} as
\begin{equation}
	L(p)=\frac{1}{\mu }\int_{0}^{p}Q(t)\,dt, \quad p\in \left[ 0,1\right].
	\label{eq:Lorenz_def}
\end{equation}
Since $Q$ is non-negative and non-decreasing, $L$ is itself non-decreasing and convex, with $L(0)=0$ and $L(1)=1$. Moreover, the presence of a point mass at zero implies $L(p)=0$ for $p \in [0,\pi]$ and positive otherwise, so that $\pi$ naturally identifies the left endpoint of the support of $L$. Furthermore, $L$ determines the distribution of $X$ up to a scale factor given by $\mu$: given $L$ and $\mu$, the quantile function is recovered as $Q(p) = \mu L'(p)$, and, provided $L$ is twice differentiable, the density of the absolutely continuous component satisfies $f(Q(p)) = (\mu \, L''(p))^{-1}$ for any $p>\pi$, so that the full distribution is uniquely determined. Conversely, without $\mu$, all random variables of the form $cX$, for $c > 0$, share the same \ac{LC} and are thus observationally equivalent from the perspective of $L$ alone (see~\citealp{bon:goi19}).

To ensure that the \ac{LC} has the analytical properties required below, introduce the following class.
\begin{definition}\label{def:Lorenz_curves}
	A map $L:[0,1]\to[0,1]$ belongs to the class $\mathcal{L}or$ if it satisfies the 
	following properties:
	\begin{enumerate}[label=\upshape{($\mathbf{\mathrm{L}_{\arabic*}}$)}, 
		ref=\upshape{($\mathbf{\mathrm{L}_{\arabic*}}$)}]
		\item \label{Lor-prop:1} $L(0)=0$ and $L(1)=1$.
		\item \label{Lor-prop:support} There exists $\pi\in(0,1)$ such that $L(p)=0$ for all 
		$p\le \pi$ and $L(p)\in (0,1]$ otherwise. The support of $L$ is defined as 
		$\mathrm{supp}(L)=(\pi,1]$.
		\item \label{Lor-prop:conv_incr} $L$ is convex and non-decreasing over $[0,1]$, 
		and strictly increasing over $[\pi,1]$.
		\item \label{Lor-prop:derivable_at_q} $L$ is twice differentiable over $(0,1]\setminus\{\pi\}$, with $L''(p)>0$ for every $p\in(\pi,1)$, and admits positive finite right derivatives of first and second order at $p=\pi$.
		\item \label{Lor-prop:derivable_at_1} $L$ has bounded and strictly positive second left derivative at $p=1$.
	\end{enumerate}
\end{definition}
Properties~\ref{Lor-prop:1}--\ref{Lor-prop:conv_incr}, as well as the differentiability part of~\ref{Lor-prop:derivable_at_q}, follow directly from definition~\eqref{eq:Lorenz_def}, the nature of $\Gamma$ and the properties of $Q$. 
The remaining conditions in~\ref{Lor-prop:derivable_at_q} and~\ref{Lor-prop:derivable_at_1} -- including the strict positivity of $L''$ on $(\pi,1)$ and at $p=1$ -- are additional regularity assumptions: convexity alone only gives $L''\ge0$ a.e., whereas the embedding constructed in Section~\ref{sec:diffeomorfismo} requires $L''(p)>0$ for \emph{every} $p\in(\pi,1]$, so that the logarithmic map $\psi_2=\ln(-\ell'')$ is well defined and continuous on all of $[0,1]$.
In particular, the positivity of the right second derivative of $L$ at $p = \pi$ ensures that the transformed curve constructed in Section~\ref{sec:diffeomorfismo} has a finite and strictly negative second derivative at the corresponding boundary point, a condition required for the logarithmic embedding introduced therein to be well defined, while the positivity of the right first derivative at $p = \pi$ prevents the transformed curve from having a vertical tangent at the same boundary point, so that both its first and second derivatives remain finite there (see the explicit limits computed in the proof of Proposition~\ref{pro:definition_L-tilde}). Finally, the boundedness of the second left derivative at $p=1$ leads to a twice differentiable function after the perspective transformation.
It is worth noting that the cases $\pi = 0$ and $\pi = 1$ are excluded because the bijection defined in Section~\ref{sec:diffeomorfismo} no longer applies in either case. Excluding $\pi = 1$ is not particularly restrictive, as the corresponding \ac{LC} is degenerate: it describes a situation in which the entire wealth is concentrated in a single individual, with all others holding nothing. Excluding $\pi = 0$ may appear more limiting, since this case is not degenerate and corresponds to positive random variables whose distribution is absolutely continuous with respect to the Lebesgue measure; for this setting, reference is made to \cite{bon:goi19}, where a dedicated embedding for such \ac{LCs} is proposed as the basis for \ac{FPCA}. That embedding, however, is not suited to the zero-inflated measures considered in the present work.

Finally, Definition~\ref{def:Lorenz_curves} implies that any \ac{LC} in the class $\mathcal{L}or$ can be associated with a measure $\Gamma \in \mathcal{M}$ with finite mean via~\eqref{eq:Lorenz_def}, up to scale.
In fact, given $\Gamma = \pi\delta_0 + (1-\pi)\lambda$, the associated \ac{LC} satisfies~\ref{Lor-prop:1}--\ref{Lor-prop:conv_incr} by construction, with $\pi = \inf\mathrm{supp}(L)$, whereas the regularity and strict-curvature conditions~\ref{Lor-prop:derivable_at_q}--\ref{Lor-prop:derivable_at_1} are additional restrictions: $\mathcal{L}or$ is thus a smooth, strictly convex subclass of the \ac{LCs} generated by $\mathcal{M}$. Henceforth, attention is restricted to elements of $\mathcal{L}or$, and the dependence on the underlying measure $\Gamma$ is left implicit.

\section{About the Diffeomorphism $\psi$}
\label{sec:diffeomorfismo}

This section derives, step by step, the diffeomorphism $\psi$ that bijectively embeds $\mathcal{L}or$ into a Hilbert space. 
The first step, presented in Subsection~\ref{sec:the-perspective-varphi}, introduces the perspective map $\varphi$, shows how it acts on \ac{LCs} and, in particular, proves that, for a given $L\in\mathcal{L}or$, $\varphi(L)$ identifies a well-defined non-negative function $\ell$ on $[0,1]$ with the required regularity properties.
The second step shows how to bijectively link $\ell$ with its intercept and its second derivative. This is done in Subsection~\ref{sec:the-BVP} through the use of a suitable Boundary Value Problem (BVP).
The last step, see Subsection~\ref{sec:the-psi}, concerns the composition of such maps that leads to the diffeomorphism $\psi$ from $\mathcal{L}or$ to $\mathbb{R} \times C([0,1])$ where the \ac{FPCA} can be carried out (see Section~\ref{sec:consistency_results}).
\\%
For readability, the main theoretical results concerning the diffeomorphism are stated in this section, while the proofs are postponed to Appendix~\ref{sec:proofs}.

\subsection{The perspective $\varphi$}
\label{sec:the-perspective-varphi}

\begin{proposition}\label{pro:perspective}
Consider the function
\begin{align*}
	&\varphi : \mathbb{R}^2 \setminus \{(x_1,x_2) : x_2 = x_1 - 1\} \to \mathbb{R}^2 \setminus \{(x_1',x_2') : x_2' = -1\}, \\
	&\qquad\qquad
	\left(\begin{array}{c}
		x_1'\\
		x_2'
	\end{array}\right)
	=
	\left(\begin{array}{c}
		x_2/(x_2-x_1+1)\\
		(x_1-x_2)/(x_2-x_1+1)
	\end{array}\right).
\end{align*}
This is a perspective map and its inverse is given by
\begin{align*}
	&\varphi^{-1} : \mathbb{R}^2 \setminus \{(x_1',x_2') : x_2' = -1\} \to \mathbb{R}^2 \setminus \{(x_1,x_2) : x_2 = x_1 - 1\}, \\
	&\qquad\qquad
	\left(\begin{array}{c}
		x_1\\
		x_2
	\end{array}\right)
	=
	\left(\begin{array}{c}
		(x_1' + x_2') / (x_2' + 1)\\
		x_1'/(x_2' + 1)
	\end{array}\right).
\end{align*}
\end{proposition}
For readability, the way such perspective is obtained together with usual geometric properties are postponed in Appendix~\ref{sec:about_varphi} whereas the proof of Proposition~\ref{pro:perspective} is in Appendix~\ref{sec:proofs}.

It is worth noting that, since $\varphi$ is a perspective mapping, $\varphi(L)$ is generally not a function but rather a parametric curve in $\mathbb{R}^2$. In particular, applying $\varphi$ to a $L\in\mathcal{L}or$ leads to a parametric expression of the transformed \ac{LC}:
\begin{equation}\label{Eq:persp}
\varphi(p, L(p)) = \left(\frac{L(p)}{L(p)-p+1} ,\frac{p-L(p)}{L(p)-p+1}\right),\qquad p\in[0,1].
\end{equation}
Next proposition shows that this parametric curve identifies a proper function $\ell$ from $[0,1]$ to $[0,\infty)$ and provides some of its basic properties. 
\begin{proposition}\label{pro:definition_L-tilde}
Let $L:[0,1]\to[0,1]$ be a \ac{LC} in $\mathcal{L}or$. Then, there exists a function $\ell:[0,1] \to [0,+\infty)$ whose graph on $[0,1]$ coincides with that of $\varphi(p,L(p))$ for $p\in[\pi,1]$ with $\pi=\inf \mathrm{supp}(L)$. Moreover, $\ell(0) = \pi(1-\pi)^{-1}$, $\ell(1) = 0$, and the function $\ell$ is concave and twice continuously differentiable on $[0,1]$.
\end{proposition}
Although Proposition~\ref{pro:definition_L-tilde} ensures that $\varphi(L)$ is in fact the graph of a function, which is denoted by $\ell$, in general, this function does not admit an explicit analytical expression. Nevertheless, its derivatives can be expressed parametrically as follows.
In fact, keeping in mind that the graph of $\ell$ coincides with that of the parametric curve given in Equation~\eqref{Eq:persp} and that $\varphi$ induces the following change of variables:
\[
	x(p) =\frac{L(p)}{L(p)-p+1},
	\qquad\qquad
	y(p) = \frac{p-L(p)}{L(p)-p+1},
\]
then, for any $p\in (\pi,1)$, it is possible to provide the derivatives of $\ell$ on $(0,1]$ depending on the derivatives of $L$ on $(\pi,1]$
\begin{align}
	\ell'(x(p)) = & \frac{dy}{dx} = \frac{ dy/dp}{dx/dp} =
	\frac{1+p-L'(p)}{L(p)+L'(p)(1-p)}, \notag
\\
	\ell''(x(p))  = & \frac{d^2y}{dx^2} = \frac{d}{dx}\left(\frac{ dy/dp}{dx/dp}\right) =
 \frac{-L''(p)(L(p)-p+1)^3}{(L(p)+L'(p)(1-p))^3}. \label{eq:ell''}
\end{align}
Note that $\ell'$ and $\ell''$ cannot be expressed in terms of $x$ and $y$ solely. Moreover, for $p\in [0,\pi]$, $L(p)=L'(p)=0$ and then $\ell'$ and $\ell''$ are not defined: in fact, as shown in the second step of the proof of Proposition~\ref{pro:definition_L-tilde}, for $p \in [0,\pi]$, the graph of  \eqref{Eq:persp} coincides with the segment $\{0\}\times [0, \pi(1-\pi)^{-1}]$ that lies on $x_1'=0$.

\subsection{The BVP}
\label{sec:the-BVP}

As shown in Proposition~\ref{pro:definition_L-tilde}, the transformed curve $\ell$ is a concave, twice continuously differentiable function on $[0,1]$ with $\ell(1)=0$ and intercept $\ell(0) = \pi(1-\pi)^{-1}$. A key observation is that $\ell$ is uniquely determined by the pair $(\ell(0), \ell'')$: indeed, given the intercept and the second derivative, $\ell$ can be recovered as the unique solution of the following Boundary Value Problem (BVP)
\begin{equation}
	\label{eq:BVP_problem}
	{\left\{ 
		\begin{array}{ll}
			u^{\prime \prime }(p) =f(p),\quad  & \mbox{if } p\in [0,1], \\
			u(0) = a, &  \\ 
			u(1) = 0, &
		\end{array}
		\right. }
\end{equation}
with $a = \ell(0)$ and $f = \ell''$. The well-posedness of this BVP is established in the following proposition.

\begin{proposition}
	\label{pro:BVP_solution}
	Let $a\ge 0$ and $f$ be a real valued continuous function defined on $[0,1]$. 
	The BVP~\eqref{eq:BVP_problem} is well-posed: its solution exists, is unique, 
	and depends continuously on the input data $a$ and $f$. Moreover, the solution 
	admits the explicit representation
	\begin{equation}
		\label{eq:BVP_solution}
		u(p) = BVP(a, f)(p) = a(1-p) + \int_0^p (p-t) f(t)\,dt - 
		p\int_0^1 (1-t) f(t)\,dt.
	\end{equation}
\end{proposition}

The explicit formula~\eqref{eq:BVP_solution} plays a central role in the construction of the inverse map $\psi^{-1}$ as can be seen in the next section.

\subsection{The bijection $\protect\psi$}
\label{sec:the-psi}   

Thanks to Proposition~\ref{pro:definition_L-tilde}, $\varphi$ maps a \ac{LC} $L$ into a smooth concave real valued function $\ell$ defined on $[0,1]$. On the other hand, $\ell$ can be seen as the solution of the BVP \eqref{eq:BVP_problem} with $f(p)=\ell''(p)$ and $\ell(0) = \pi(1-\pi)^{-1}$ where $\pi=\inf \mathrm{supp}(L)$. 
This means that each \ac{LC} $L\in \mathcal{L}or$ can be univocally mapped into the pair $(\ell(0),\ell'')$. Since $L''(p)>0$ for every $p\in(\pi,1]$ (see~\ref{Lor-prop:derivable_at_q}--\ref{Lor-prop:derivable_at_1}), Equation~\eqref{eq:ell''} gives $\ell''(x)<0$ strictly for every $x\in[0,1]$; since also $\ell(0)>0$, $L$ can be mapped into the pair $(\ln(\ell(0)),\ln(-\ell'')) \in \mathbb{R}\times C([0,1])$.
\\%
The following proposition establishes the formal properties of this map.
\begin{proposition}
	\label{pro:psi_bi-continuous_bijection}
	Let $\mathcal{L}or$ be endowed with the sup-norm metric $\|\cdot\|_\infty$ and let $\mathbb{R} \times C([0,1])$ be endowed with the product metric $d(a,g)=\lvert a\rvert +\| g \|_\infty$. Then the map 
	\begin{equation}
		\label{eq:bijection_psi}
		\begin{array}{rccl}
			\psi : & \mathcal{L}or & \to & \mathbb{R}\times C([0,1])
			\\
			& L & \mapsto & \psi(L)= (\psi_1(L),\psi_2(L) )=(\ln(\ell(0)),\ln(-\ell''))
		\end{array}
	\end{equation}
	is a bijection from $\mathcal{L}or$ onto $\mathbb{R}\times C([0,1])$. Moreover, the inverse map is explicitly defined on the whole space $\mathbb{R}\times C([0,1])$ by
	\[
	\psi^{-1}(b,g) = \varphi^{-1}\left[ BVP( \exp(b), -\exp(g) ) \right], \qquad \forall (b,g)\in \mathbb{R}\times C([0,1]).
	\]
	As shown in the proof of this proposition (see Appendix~\ref{sec:proofs}), this expression always belongs to $\mathcal{L}or$, so that $\psi^{-1}$ is well defined as a map from $\mathbb{R}\times C([0,1])$ to $\mathcal{L}or$.
	Moreover, the maps involved in the construction of $\psi^{-1}$, namely
	\[
	(b,g)\mapsto (e^b,-e^g), \qquad
	(a,f)\mapsto BVP(a,f), \qquad
	u\mapsto \phi^{-1}(\operatorname{graph} u),
	\]
	are smooth on their natural domains. Consequently, after equipping $\mathcal{L}or$ with the smooth structure pulled back through $\psi$ from $\mathbb{R}\times C([0,1])$, the map
	\[
	\psi:\mathcal{L}or\longrightarrow \mathbb{R}\times C([0,1])
	\]
	is a diffeomorphism.
\end{proposition}
In other words, given a pair $(b,g) \in \mathbb{R}\times C([0,1])$, the inverse $\psi^{-1}$ proceeds in three steps: first, the intercept and second derivative of $\ell$ are recovered as $\ell(0) = \exp(b)$ and $\ell'' = -\exp(g)$; second, $\ell$ is reconstructed via the BVP solution $BVP(\exp(b), -\exp(g))$, which is well-posed by Proposition~\ref{pro:BVP_solution}; third, $\varphi^{-1}$ maps $\ell$ back to a \ac{LC} $L \in \mathcal{L}or$.

Mathematically, since $\psi$ is a diffeomorphism, one can equivalently work with a \ac{LC} $L\in \mathcal{L}or$ or with its image $\psi(L)\in \mathbb{R} \times C([0,1])$, and all convergence results established in $\mathbb{R} \times C([0,1])$ (see Propositions~\ref{pro:cons-MoV-full-observed-LCs} and~\ref{pro:cons-est-MoV}) can be continuously pulled back to $\mathcal{L}or$, ensuring full interpretability of the results. Moreover, the perspective $\varphi$ acts as a nonlinear transformation that simultaneously reparametrizes the domain of $L$ from $(\pi,1]$ to $[0,1]$ and maps its values from $[0,1]$ to $[0,+\infty)$, in the spirit of curve registration methods in FDA \citep{ram:sil05}: the scalar component $\psi_1(L) = \ln(\pi/(1-\pi))$ captures the variability of the discrete part of $\Gamma$, while the functional component $\psi_2(L) = \ln(-\ell'')$ captures that of the absolutely continuous part.

Finally, this transformation also has an interpretive role, in analogy with the interpretation of the second derivative of a \ac{LC}, $L''(p)$ as a local inequality weight at the $p$-th quantile (see \cite{bon:goi19}), the second derivative $\ell''(x)$ for $x\in(0,1)$ admits a natural interpretation as the local radial weight of inequality along the direction $y = -x(x-1)^{-1}(p-1)$, i.e.\@ the straight line passing through $(1,0)$ with slope $-x(x-1)^{-1}$ and $y$-intercept $x(x-1)^{-1}$ (the odds ratio of $x$).
The point \((1,0)\) therefore provides a natural pivot for describing radial deviations from the equality line under the perspective transformation: it is the centre of projection, and the pencil of lines through it -- mapped by $\varphi$ into the vertical lines along which $\ell''$ is read -- gives the radial directions along which inequality is measured. Loosely speaking, $(1,0)$ plays the role of a centre of attraction pulling the curve away from perfect equality, the local weight $-\ell''(x)$ measuring the strength of this pull in the direction $x$; the analogy is of course only suggestive, and the formal content is the projective structure just described.

\section{Modes of Variations and asymptotics}
\label{sec:consistency_results}

This section presents the main statistical results and is divided into three parts. In the first (and present) part the notion of random \ac{LC} is introduced together with two main assumptions. In the second part (see~Subsection~\ref{sec:covariance&modes}) the covariance operator, the modes of variations of a random \ac{LC} and their empirical counterparts are introduced. The main results in this subsection include some properties for the covariance operator and the $L^2(\mathbb{P})$ and a.s.\@ consistency for the modes of variations. Finally, Subsection~\ref{sec:double_stoc} discusses the problem of double stochasticity which is anticipated in the introduction and arises whenever each \ac{LC} realization must be estimated from a sample of the underlying random variable.

Thanks to the diffeomorphism $\psi$ introduced in the previous section, the family of \ac{LCs} $\mathcal{L}or$ can be bijectively embedded into $\mathbb{B}=\mathbb{R}\times C([0,1])$, where $C([0,1])$ denotes the family of continuous real--valued functions defined on $[0,1]$. If endowed with the norm $\vert b \vert + \Vert g \Vert_\infty$, $\mathbb{B}$ is a Banach space. Moreover, it is a dense subset of $\mathbb{H} = \mathbb{R}\times L^2_{[0,1]} $ where $L^2_{[0,1]} $ denotes the family of square integrable real--valued functions defined on $[0,1]$ which is a separable Hilbert space whenever it is endowed with the scalar product $\langle (b_1,g_1), (b_2,g_2)\rangle_\mathbb{H} = b_1b_2 + \langle g_1, g_2\rangle_{L^2_{[0,1]} }$. Thus, $\mathcal{L}or$ can exploit the structure of $\mathbb{H}$ in defining the objects necessary to address the \ac{FPCA} problem.
\\
In this view, assume that $(\Omega, \mathcal{A}, \mathbb{P})$ is a probability space, then $L$ is a random element in $\mathcal{L}or$ if its image $\xi=\psi(L)$ is a random element in $(\mathbb{H},\mathcal{F})$ where $\mathcal{F}$ is the $\sigma$--algebra generated by the topology induced by the mentioned scalar product. Moreover, consider the following moment assumptions on $L$.
\begin{assume}\label{hp:momento_4_q}
	The parameter $\pi$ in~\ref{Lor-prop:support}, viewed as a random variable, 
	satisfies 
	\[
	\mathbb{E}\!\left[\ln^4\!\left(\frac{\pi}{1-\pi}\right)\right] < \infty.
	\]
\end{assume}
\begin{assume}\label{hp:momento_4_L}
	The second derivative of $\ell$ satisfies
	\[
	\mathbb{E}\left[ \sup_{x\in[0,1]} \left\vert \ln\left(-\ell''\left(x\right)\right) \right\vert^4 \right] < \infty.
	\]
\end{assume}

Assumption~\ref{hp:momento_4_q} is a moment condition on the log-odds of the structural zero mass, satisfied whenever $\pi$ does not concentrate too close to $0$ or $1$ in the population. Assumption~\ref{hp:momento_4_L} is a uniform moment condition on the second derivative of $L$: recalling from Equation~\eqref{eq:ell''} that, for $x=x(p)$,
\[
-\ell''(x(p)) = \frac{L''(p)\left(L(p)-p+1\right)^3}{\left(L(p)+L'(p)(1-p)\right)^3}, \qquad p\in\left[\pi,1\right],
\]
the quantity inside the supremum measures, at each $p$, the curvature of $L$ relative to its own slope, weighted by a term that depends on the geometry of the support $\left[\pi,1\right]$ and degenerates near its endpoints, where $L$ approaches its flat segment on one side and its slope at $p=1$ on the other. Assumption~\ref{hp:momento_4_L} therefore requires that this relative curvature does not go to infinity, in a fourth-moment sense, anywhere over the support of $L$ -- including regions where $L''$ may become very small (near $p=\pi$, where the \ac{LC} straightens out) or where it may grow large (near $p=1$) -- providing a uniform counterpart to the pointwise control already implied by~\ref{Lor-prop:derivable_at_q}--\ref{Lor-prop:derivable_at_1}. Both assumptions are tailored to the specific structure of the embedding defined in Section~\ref{sec:diffeomorfismo} and together guarantee the existence of the covariance operator and allow the almost sure and $L^2(\mathbb{P})$ consistency results of~\cite{bos00,hor:kok12} to be applied.

\subsection{Covariance Structure and Empirical Modes of Variations}
\label{sec:covariance&modes}

Since $\xi = (\xi^{(1)},\xi^{(2)}) = (\ln(\ell(0)),\ln(-\ell''))$ where $\xi^{(1)}$ is a real random variable and $\xi^{(2)}$ a random process in $C([0,1])$, its covariance operator can be written using a block matrix structure in the following way
\begin{equation}
	\Sigma_\xi=\left[
	\begin{array}{cc}
		\sigma^2_{\xi^{(1)}}	& \Sigma_{\xi^{(1)}, \xi^{(2)}}
		\\ 
		\Sigma_{\xi^{(2)}, \xi^{(1)}} & \Sigma_{\xi^{(2)}}
	\end{array}
	\right],
	\label{eq:block_theoretical_cov}
\end{equation}
where, for any $f\in L^2_{[0,1]}$ and $x\in [0,1]$,
\begin{align*}
	\sigma^2_{\xi^{(1)}} & =
	\mathbb{E}\left[ \left( \xi^{(1)} - \mathbb{E}\left[\xi^{(1)} \right] \right)^2 \right],
	\\
	\Sigma_{\xi^{(2)}} (f) &= \mathbb{E}\left[ \left( \xi^{(2)} - \mathbb{E}\left[\xi^{(2)} \right] \right) \left\langle \xi^{(2)} - \mathbb{E}\left[\xi^{(2)}\right] , f \right\rangle_{L^2_{[0,1]}} \right],
	\\
	\Sigma_{\xi^{(1)}, \xi^{(2)}}(x) & = \Sigma_{\xi^{(2)}, \xi^{(1)}}(x) = \mathbb{E}\left[ \left( \xi^{(1)} - \mathbb{E}\left[\xi^{(1)} \right] \right) \left( \xi^{(2)}(x) - \mathbb{E}\left[\xi^{(2)}(x)\right] \right) \right]
	\\
	& = Cov\left(\xi^{(1)} , \xi^{(2)}(x)\right).	
\end{align*}
Moreover, $Cov(\xi^{(2)}(x_1), \xi^{(2)}(x_2))$ is the covariance function of $\Sigma_{\xi^{(2)}}$ where $x_1, x_2\in [0,1]$. 
Next proposition establishes some useful properties of $\Sigma_\xi$, in particular, it shows that the second component of the eigenvector associated to such operator is continuous. 
\begin{proposition}\label{pro:eigenfunction_continuous}
	Consider $\Sigma_{\xi}$ as in Equation~\eqref{eq:block_theoretical_cov} and assume \ref{hp:momento_4_q} and \ref{hp:momento_4_L}, then the covariance $Cov(\xi^{(2)}(x_1), \xi^{(2)}(x_2))$ and the cross-covariance $Cov(\xi^{(1)}, \xi^{(2)}(x))$ are continuous functions for any $x,x_1,x_2 \in [0,1]$. 
	\\
	Moreover, if $\nu\in \mathbb{H}= \mathbb{R} \times L^2_{[0,1]}$ is an eigenvector of $\Sigma_{\xi}$ associated with eigenvalue $\alpha$, then $\nu\in \mathbb{B}=\mathbb{R}\times C([0,1])$.
\end{proposition}

Denote by $\{\alpha_j, \nu_j\}$ the eigenvalues and eigenvectors of $\Sigma_\xi$ with $\alpha_j > \alpha_{j+1}>0$ for any $j\in \mathbb{N}$. It is assumed that $\alpha_j$ is simple and well separated from its neighbours, i.e.\@
\[
\delta_j := \min\left\{\alpha_{j-1}-\alpha_j,\ \alpha_j-\alpha_{j+1}\right\} > 0,
\]
with the convention $\delta_1 := \alpha_1-\alpha_2$ when $j=1$.
Within the framework of PCA, the eigenvectors $\nu_{j}$ are defined as the directions in $\mathbb{H}$ that sequentially maximize the variance of the process $\xi $. Specifically, $\nu_{j}$ is the direction along which $\xi $ exhibits the largest possible variance $\alpha_{j}$, subject to being orthogonal to the previous $(j-1)$ eigenvectors.
From this perspective, to qualitatively assess the variability of $\xi$ along such a direction, it is useful to define and visualize the $j$-th theoretical Mode of Variation (th-MoV) (see, e.g., \citealp[Section~8.3]{ram:sil05}):
\begin{equation}\label{eq:xi_theor_MoV}
	m_{j}^k(x) = \mu(x) \pm k \sqrt{\alpha_j} \nu_j(x), \qquad x\in [0,1], %
\end{equation}
where $\mu = \mathbb{E}[\xi]$ and $k>0$, typically chosen as $k\in \{1,2,3\}$, represents the number of standard deviations from the mean. Note that, since $\xi \in \mathbb{B} = \mathbb{R}\times C([0,1])$, the dependence on $x$ in \eqref{eq:xi_theor_MoV} pertains exclusively to the second (functional) component of the elements involved. 
In practice, these modes are usually investigated only for the first few eigenfunctions (typically $j=1,\dots ,J$ with $J\in \{3,4,5\}$). These directions do not only account for most of the process variability, but they also correspond to those characterized by the lowest estimation error (see e.g.\@ \citealp{hal:hos06,hor:kok12}).
Proposition~\ref{pro:eigenfunction_continuous} implies that $\nu _{j}\in \mathbb{B}$, ensuring that the th-MoVs are elements of $\mathbb{B}$ as well. This fact allows the Mode of Variation for the \ac{LC} (LC-MoV) to be defined as the counterimage of $\psi$ as follows
\begin{equation*}\label{eq:LC_theor_MoV}
	M_{j}^k = \psi^{-1}\left(m_{j}^k\right), %
\end{equation*}

As anticipated in the Introduction, the objective is to provide an inferential framework for consistently estimating $M_{j}^k$.
To do this, consider a sample of (fully observed) random \ac{LCs}, that is $L_1,\ldots, L_n$ i.i.d.\@ as $L$ in $\mathcal{L}or$, and the corresponding sample $\xi_1=\psi(L_1),\ldots, \xi_n=\psi(L_n)$ from which it is possible to compute the empirical covariance operator $\Sigma_{\xi,n}$ and its eigenelements $\{\alpha_{j,n}, \nu_{j,n} \}_{j=1}^n$. Consequently, the $j$--th Mode of Variation estimated from a sample of fully observed \ac{LCs} (fo-MoV) is
\begin{equation*}\label{eq:xi_fo_MoV}
	m_{j,n}^k = \mu_n \pm k \sqrt{\alpha_{j,n}} \nu_{j,n}, %
\end{equation*}
where $1\le j\le n$, $k>0$ and $\mu_n = n^{-1}\sum_{i=1}^{n}\xi_i$. 
Moreover, since $\Sigma_{\xi,n}$ is a covariance operator defined on $\mathbb{H}$, the results of Proposition~\ref{pro:eigenfunction_continuous} can be adapted to show that the empirical eigenfunctions $\nu_{j,n}$ also belong to $\mathbb{B}$. 
Consequently, the fo-MoVs are elements of the same space, so that the Mode of Variation for the \ac{LC} (LC-fo-MoV) can be defined through the inverse mapping $\psi ^{-1}$ as follows:
\begin{equation*}\label{eq:LC_fo_MoV}
	M_{j,n}^k = \psi^{-1}\left(m_{j,n}^k \right).%
\end{equation*}
In this setting, the following proposition holds true (see Appendix~\ref{sec:proofs} for the proof).
\begin{proposition}\label{pro:cons-MoV-full-observed-LCs}
	Under assumptions \ref{hp:momento_4_q} and \ref{hp:momento_4_L}, for a given $k\ge 0$ and as $n \to \infty$,
	\[
	\left\Vert m_{j,n}^k - m_{j}^k \right\Vert_\mathbb{B} \to 0, \qquad
	\left\Vert M_{j,n}^k - M_{j}^k \right\Vert_\infty \to 0 \qquad a.s..
	\]
	Moreover, the convergence is also in $L^2(\mathbb{P})$, that is, as $n\to \infty$,
	\[
	\mathbb{E}\left[ \left\Vert m_{j,n}^k - m_{j}^k \right\Vert_\mathbb{B}^2 \right] \to 0,
	\qquad
	\mathbb{E}\left[ \left\Vert M_{j,n}^k - M_{j}^k \right\Vert_\infty^2 \right] \to 0.
	\]
\end{proposition}

Proposition~\ref{pro:cons-MoV-full-observed-LCs} implicitly guarantees that $\mu_n, \alpha_{j,n}, \nu_{j,n}$ are a.s.\@ (and $L^2(\mathbb{P})$) consistent estimators for $\mu, \alpha_{j}, \nu_{j}$.

\subsection{Dealing with Double Stochasticity}
\label{sec:double_stoc}

Proposition~\ref{pro:cons-MoV-full-observed-LCs} provides an inferential result for the MoV assuming fully observed \ac{LCs}. However, as anticipated in the Introduction, practical settings typically involve a double source of randomness. First, the theoretical \ac{LCs} $L_1, \ldots, L_n$ are sampled i.i.d.\@ as $L$ from $\mathcal{L}or$, each associated with a non-negative random variable $X_i$ drawn from a population $\mathcal{X}$. Second, these curves are not directly observed; instead, for each $X_i$, a random sample $(X_i^1, \ldots, X_i^{N_i})$ is drawn, with $N_{\min} := \min_{1\le i \le n} N_i$ and its empirical zero-mass proportion is estimated by
\[
\widehat{\pi}_i := \frac{1}{N_i}\sum_{m=1}^{N_i} \mathbb{I}\{X_i^m=0\}.
\]
From these samples, empirical \ac{LCs} $\widetilde{L}_1,\ldots,\widetilde{L}_n$ are computed via Equation~\eqref{eq:Lorenz_def} using empirical means and quantile functions.
These empirical curves are piecewise linear and therefore do not belong to $\mathcal{L}or$. Crucially, because the use of the mapping $\psi$ requires the second derivatives of $\ell = \varphi(L)$, these raw empirical curves cannot be used directly. In practice, the most convenient approach consists of first applying $\varphi$ to the piecewise linear empirical curves $\widetilde{L}_1,\ldots,\widetilde{L}_n$, and subsequently applying a smoothing procedure. This yields the smooth estimators $\widehat{\ell}_1, \ldots, \widehat{\ell}_n \in \mathbb{B}$. This smoothing step can be performed using techniques such as finite difference methods \cite{qua:sac:sal06}, constrained penalized splines \cite{mey12}, or local polynomial smoothing \cite{fan:gij96}, which produce estimators that satisfy the necessary regularity conditions.
From the given procedure, it is then possible to estimate the mean, the covariance operator and its eigenelements by $\widehat{\mu}_n$, $\widehat{\Sigma}_{\xi ,n}$ and $\{\widehat{\alpha}_{j,n},\widehat{\nu}_{j,n}\}$ respectively. From these objects, the estimated $j$--th Mode of Variation \ac{LCs} (est-MoV) can be written as
\begin{equation*}\label{eq:xi_est_MoV}
	\widehat{m}_{j,n}^k =\widehat{\mu}_n \pm k \sqrt{\widehat{\alpha}_{j,n}}\widehat{\nu}_{j,n}, \qquad \textnormal{(est-MoV)},
\end{equation*}
where $k>0$ and $\widehat{\mu}_n = n^{-1}\sum_{i=1}^{n}\widehat{\xi}_i$. With analogous arguments as above, also the est-MoVs are elements in $\mathbb{B}$ and it is possible to define the Mode of Variation for the \ac{LC} using the inverse mapping $\psi^{-1}$ as follows
\begin{equation*}\label{eq:LC_est_MoV}
	\widehat{M}_{j,n}^k = \psi^{-1}(\widehat{m}_{j,n}^k), \qquad \textnormal{(LC-est-MoV)}.
\end{equation*}
The following result is stated under a high-level assumption on the smooth estimators, formulated directly in the embedded space $\mathbb{B}$, where the proof operates. Verifying this assumption for a specific smoothing procedure is estimator-dependent and is not pursued here in full generality; sufficient conditions are discussed in Remark~\ref{rem:sufficient_A4} below. Since $N_{\min}$ grows with $n$, the estimators form a triangular array, $\widehat{\xi}_i=\widehat{\xi}_{n,i}$, $i=1,\ldots,n$, whereas $\xi_1,\xi_2,\ldots$ are the first elements of a fixed i.i.d.\@ sequence as in Proposition~\ref{pro:cons-MoV-full-observed-LCs}; the dependence on $n$ is omitted when no confusion arises.

\begin{assume}\label{hp:estimator_moment}
	For every $n$ and $i=1,\ldots,n$, the smooth estimator $\widehat{L}_i$ belongs to $\mathcal{L}or$ a.s., so that $\widehat{\xi}_i=\psi(\widehat{L}_i)\in\mathbb{B}$, and
	\begin{enumerate}[label=(\roman*), ref=(\roman*)]
		\item  \label{ass:moment_4_hat} 
		$\displaystyle \sup_{n\ge1}\max_{1\le i\le n}\mathbb{E}\left[\left\Vert\widehat{\xi}_i\right\Vert_\mathbb{B}^4\right]<\infty$;
		\item \label{ass:D_n}
		$\displaystyle D_n:=\frac1n\sum_{i=1}^n\left\Vert\widehat{\xi}_i-\xi_i\right\Vert_\mathbb{B}^2\to0$ a.s.\@ as $n\to\infty$, with $N_{\min}=N_{\min}(n)\to\infty$.
	\end{enumerate}
\end{assume}

\begin{remark}\label{rem:sufficient_A4}
	Assumption~\ref{hp:estimator_moment} separates the statistical argument from the analysis of a specific smoother. Condition~(ii) holds, in particular, whenever $\max_{1\le i\le n}\Vert\widehat{\xi}_i-\xi_i\Vert_\mathbb{B}\to0$ a.s., since $D_n\le\max_{1\le i\le n}\Vert\widehat{\xi}_i-\xi_i\Vert_\mathbb{B}^2$. This uniform version is in turn implied by the following two conditions on the scalar and on the functional component, where $\widehat{\ell}_i=\varphi(\widehat{L}_i)$ and the smoother is constrained so that $\widehat{\ell}_i(0)=\widehat{\pi}_i/(1-\widehat{\pi}_i)$.
	\begin{enumerate}[label=(\alph*)]
		\item \emph{Scalar component.} Suppose that $\pi_i\in[\eta,1-\eta]$ a.s.\@ for some $\eta\in(0,1/2)$ and that $N_{\min}/\log n\to\infty$. Conditionally on $\pi_i$, $\widehat{\pi}_i$ is the average of $N_i$ i.i.d.\@ Bernoulli$(\pi_i)$ variables, so Hoeffding's inequality \citep{hoe63} and the union bound give, for every $\varepsilon>0$,
		\[
		\mathbb{P}\left(\max_{1\le i\le n}\vert\widehat{\pi}_i-\pi_i\vert\ge\varepsilon\right)\le 2n\exp\left(-2N_{\min}\varepsilon^2\right)\le 2n^{-3},
		\]
		as soon as $N_{\min}\ge2\varepsilon^{-2}\log n$, which holds for $n$ large enough. By the Borel--Cantelli Lemma \citep[pag.~255]{shi2016}, applied for every rational $\varepsilon>0$, $\max_{i\le n}\vert\widehat{\pi}_i-\pi_i\vert\to0$ a.s.. In particular, eventually $\widehat{\pi}_i\in[\eta/2,1-\eta/2]$ for all $i\le n$; on this interval $t\mapsto\ln(t/(1-t))$ is Lipschitz with constant $C_\eta=4/(\eta(2-\eta))$, so that $\max_{i\le n}\vert\widehat{\xi}^{(1)}_i-\xi^{(1)}_i\vert\le C_\eta\max_{i\le n}\vert\widehat{\pi}_i-\pi_i\vert\to0$ a.s..
		\item \emph{Functional component.} Suppose that $-\ell_i''\ge\kappa$ on $[0,1]$ a.s.\@ for some $\kappa>0$ and all $i$, and that $\max_{1\le i\le n}\Vert\widehat{\ell}_i''-\ell_i''\Vert_\infty\to0$ a.s.. Then, eventually, $\Vert\widehat{\ell}_i''-\ell_i''\Vert_\infty\le\kappa/2$ for all $i\le n$, hence $-\widehat{\ell}_i''\ge\kappa/2$ and, since $\vert\ln s-\ln t\vert\le\vert s-t\vert/\min\{s,t\}$ for $s,t>0$, $\max_{i\le n}\Vert\widehat{\xi}^{(2)}_i-\xi^{(2)}_i\Vert_\infty\le(2/\kappa)\max_{i\le n}\Vert\widehat{\ell}_i''-\ell_i''\Vert_\infty\to0$ a.s..
	\end{enumerate}
	Condition~(b) is formulated on $\ell_i$ rather than on $L_i$: $\ell_i$ and $\widehat{\ell}_i$ share the domain $[0,1]$, which avoids comparing curves with different supports $(\pi_i,1]$ and $(\widehat{\pi}_i,1]$, and this matches the computational pipeline of Section~\ref{sec:comp_notes}, where smoothing is performed after the perspective transformation. Results for shape-constrained regression splines \citep{mey:kim:wan18,sch07} provide convergence rates of this type under smoothness assumptions on the underlying curves (e.g.\@ four continuous derivatives), but only in $L^2$, in probability, and for a fixed curve; a uniform-in-$i$, sup-norm, almost sure version is not available and is left for future work.
	Regarding condition~(i), note that, for finite $N_i$, $\widehat{\pi}_i\in\{0,1\}$ with positive probability, so $\widehat{\pi}_i$ must be kept away from the boundary in practice, e.g.\@ by truncation to $[1/(2N_i),1-1/(2N_i)]$. Under the setting of~(a), this affects $\widehat{\xi}^{(1)}_i$ only on an event of probability at most $2(1-\eta)^{N_i}$, on which $\vert\widehat{\xi}^{(1)}_i\vert\le\ln(2N_i)$, so that the fourth moment remains uniformly bounded. More generally, condition~(i) holds whenever $\Vert\widehat{\xi}_i-\xi_i\Vert_\mathbb{B}\le K$ a.s.\@ for some constant $K$, since then $\Vert\widehat{\xi}_i\Vert_\mathbb{B}^4\le8(\Vert\xi_i\Vert_\mathbb{B}^4+K^4)$ and Lemma~\ref{lem:boundedness_fourth_moment} applies.
\end{remark}

\begin{proposition}\label{pro:cons-est-MoV}
	Under assumptions \ref{hp:momento_4_q}--\ref{hp:estimator_moment}, for a given $k\ge0$ and as $n\to\infty$,
	\[
	\left\Vert \widehat{m}_{j,n}^k - m_{j}^k \right\Vert_\mathbb{B} \to 0, \qquad
	\left\Vert \widehat{M}_{j,n}^k - M_{j}^k \right\Vert_\infty \to 0 \qquad a.s..
	\]
	Moreover, the convergence is also in $L^2(\mathbb{P})$, that is, as $n\to\infty$,
	\[
	\mathbb{E}\left[ \left\Vert \widehat{m}_{j,n}^k - m_{j}^k \right\Vert_\mathbb{B}^2 \right] \to 0,
	\qquad
	\mathbb{E}\left[ \left\Vert \widehat{M}_{j,n}^k - M_{j}^k \right\Vert_\infty^2 \right] \to 0.
	\]
\end{proposition}

Proposition~\ref{pro:cons-est-MoV} implicitly guarantees that  $\widehat{\mu}_n, \widehat{\alpha}_{j,n}, \widehat{\nu}_{j,n}$ are a.s.\@ (and $L^2(\mathbb{P})$) consistent estimators for $\mu, \alpha_{j}, \nu_{j}$. 

A further consequence of the diffeomorphism property of $\psi$ is that the truncated Karhunen--Lo\`eve reconstruction of the \ac{LC}, obtained by applying $\psi^{-1}$ to the estimated expansion
\[
\widehat{\xi}_{i,k} = \widehat{\mu}_n + \sum_{j=1}^k \langle \widehat{\xi}_i 
- \widehat{\mu}_n, \widehat{\nu}_{j,n} \rangle_\mathbb{H} \widehat{\nu}_{j,n},
\]
is consistent with its population counterpart based on the true eigenelements $\{\alpha_j, \nu_j\}$. Specifically, for a fixed curve $i$, $\Vert\psi^{-1}(\widehat{\xi}_{i,k}) - \psi^{-1}(\xi_{i,k})\Vert_\infty \to 0$ as $n \to \infty$, as a direct consequence of the continuity of $\psi^{-1}$ and of the consistency of $\widehat{\mu}_n$ and $\widehat{\nu}_{j,n}$ established in Propositions~\ref{pro:cons-MoV-full-observed-LCs} and~\ref{pro:cons-est-MoV}, provided that $\Vert\widehat{\xi}_i-\xi_i\Vert_\mathbb{B}\to0$ a.s.\@ for that curve, as granted by the uniform version of Assumption~\ref{hp:estimator_moment}\ref{ass:D_n} discussed in Remark~\ref{rem:sufficient_A4}.

A natural extension of the present framework consists in incorporating the mean $\mu$ into the analysis. Since $\mu > 0$ by definition of $\Gamma \in \mathcal{M}$, one can augment the embedding space from $\mathbb{R} \times C([0,1])$ to $\mathbb{R}^2 \times C([0,1])$ by adding $\ln(\mu)$ as a further scalar component, yielding the embedding
\[
\widetilde{\psi}(L, \mu) = (\ln(\mu),\, \ln(\ell(0)),\, \ln(-\ell'')) 
\in \mathbb{R}^2 \times C([0,1]).
\]
The covariance operator of the augmented process retains a block structure, with an additional scalar block and two cross-covariance terms involving $\ln(\mu)$. Under the additional assumption $\mathbb{E}[\ln^4(\mu)] < \infty$, the theoretical results of the present paper extend to this setting with minor modifications, allowing the variability of the scale of the distribution to be captured jointly with that of its shape and zero mass.

\section{Empirical Application: Global Export Inequality}
\label{sec:application}

To demonstrate the empirical performance of the proposed framework, the joint \ac{FPCA} has been applied to the BACI international trade database, which provides data on bilateral trade flows for more than 200 countries at the product level \citep{CEPII:2010-23}. The analysis focuses on the export distributions of 226 countries for the year 2024 (updated as of 31 January 2026): for each country, $X$ denotes the value, in US dollars, of each exported product, so that for each of the 5606 products in the database, $X$ records the corresponding export value. For the $i$-th country, the associated \ac{LC} $L_i$ is estimated from the empirical distribution of export values via Equation~\eqref{eq:Lorenz_def}.
\\%
Structural zeros arise naturally in this context: it is well established in international economics that no country exports every product category, as trade frictions and firm heterogeneity imply that a substantial fraction of bilateral trade flows are exactly zero \citep{hel:mel:rub08}. Consequently, each country in the dataset exhibits a positive proportion of zero export values, and the corresponding \ac{LC} $L_i$ has a flat segment of length $\pi_i > 0$. The empirical range of $\pi_i$ is wide: China, the most diversified exporter in the sample, has an estimated $\pi$ of approximately $0.18\%$, while Palau %
has an estimated $\pi$ of approximately $97.72\%$, reflecting an extremely concentrated export structure.
\\%
Two approaches to \ac{FPCA} are compared on this dataset. The first, briefly described in Section~\ref{sec:naive_fpca}, is the naive approach that applies FDA techniques directly to the estimated \ac{LCs} without any preliminary transformation. The second, presented in Section~\ref{sec:FPCA&modes}, applies \ac{FPCA} after the embedding $\psi$. The computational aspects of the latter are discussed in Section~\ref{sec:comp_notes}.

\subsection{(Limitations of a) Naive \ac{FPCA}}
\label{sec:naive_fpca}

A naive approach to performing \ac{FPCA} on a sample of \ac{LCs} $L_1, \ldots, L_n \in \mathcal{L}or$ consists in applying standard FDA techniques directly in $L^2([0,1])$, without any preliminary transformation. This section illustrates why such an approach is problematic in the zero-inflated setting from the interpretation point of view.
In particular, for this naive framework, the definition of random \ac{LC} simplifies to a measurable map $L : (\Omega, \mathcal{F}, \mathbb{P}) \to (\mathcal{L}or, \mathcal{B}_{\mathcal{L}or})$, where $\mathcal{B}_{\mathcal{L}or}$ is the Borel $\sigma$-algebra on $\mathcal{L}or$ induced by the $L^2_{[0,1]}$ topology. 
Thus, since $\mathcal{L}or$ is not a linear subspace of $L^2_{[0,1]}$, linear combinations of \ac{LCs} do not generally belong to $\mathcal{L}or$: both the modes of variation and the truncated Karhunen--Lo\`eve reconstructions, defined in a standard way in terms of the mean curve, eigenvalues, and eigenfunctions of the covariance operator of $L$ in $L^2_{[0,1]}$, may fail to be convex, non-negative, or increasing over their support, and therefore cannot be interpreted as \ac{LCs} (see \citealp{bon:goi19}). This is illustrated empirically in Figure~\ref{fig:naive} (top-right and bottom-left panels), which shows examples of a reconstructed curve that lie outside $\mathcal{L}or$ and the first modes of variation respectively.
The extent of the problem can be further quantified by inspecting the factor plane spanned by the first two principal components: not every point in this plane corresponds to a valid \ac{LC}, and the region of the factor plane that produces valid reconstructions in $\mathcal{L}or$ is a strict subset of $\mathbb{R}^2$. As shown in Figure~\ref{fig:naive} (bottom-right panel), a non-negligible fraction of the empirical scores falls outside this feasible region, confirming that naive \ac{FPCA} produces invalid reconstructions for a substantial portion of the sample.
Further theoretical and descriptive analysis of the naive case -- including the observation that the embedding of \cite{bon:goi19}, based on a log-transformation of $L''$, implicitly requires $L'' > 0$ on $(0,1)$, a condition that fails for zero-inflated \ac{LCs} where $L''$ vanishes on $[0,\pi]$ -- has been carried out in \cite{bon:goi19}.
\begin{figure}[ht!]
	\centering
	\includegraphics[width=.45\textwidth]{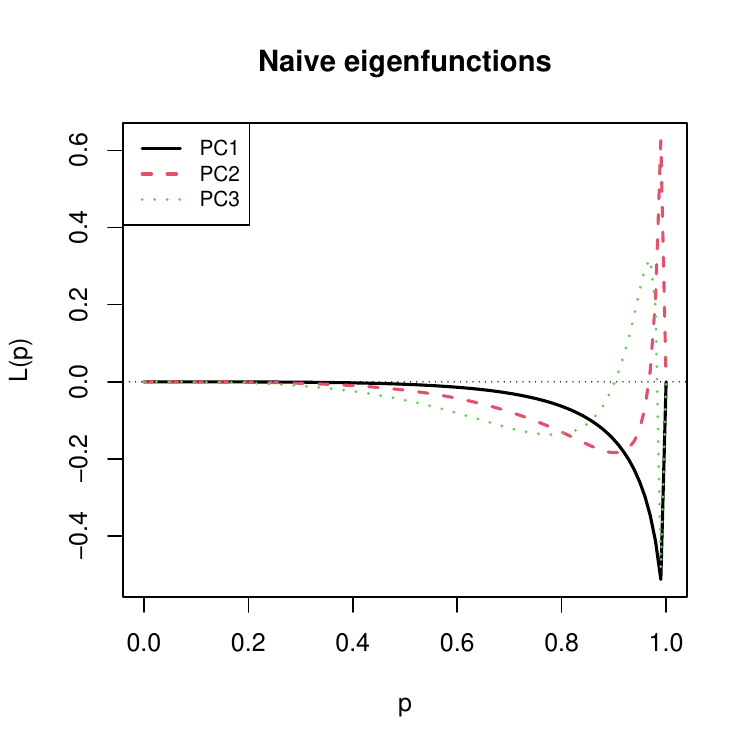}	
	\includegraphics[width=.45\textwidth]{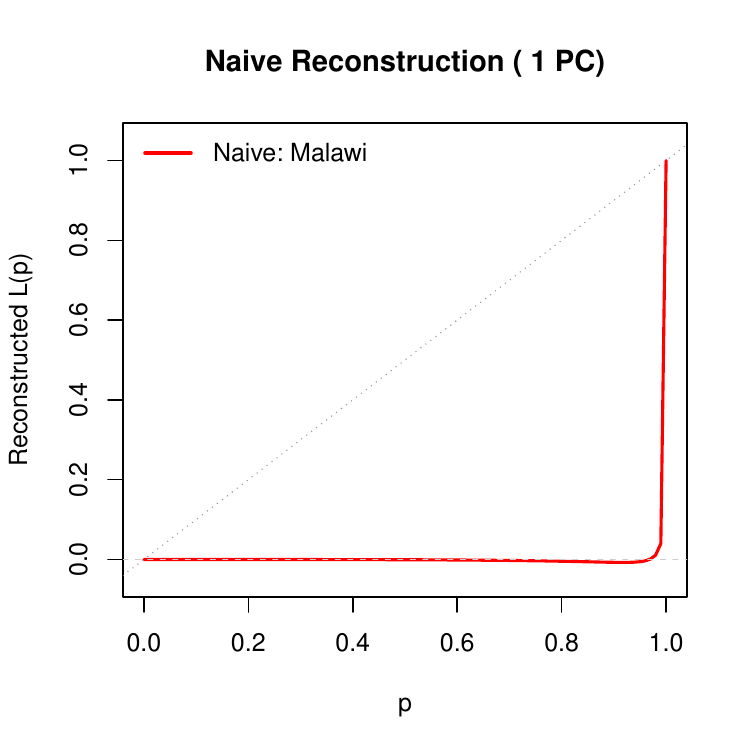}	
	\\
	\includegraphics[width=.45\textwidth]{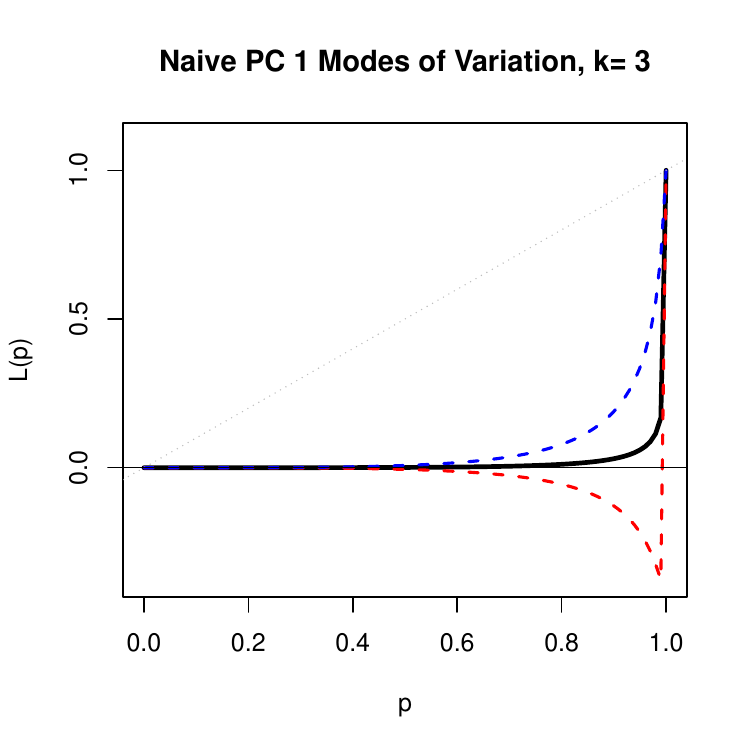}
	\includegraphics[width=.45\textwidth]{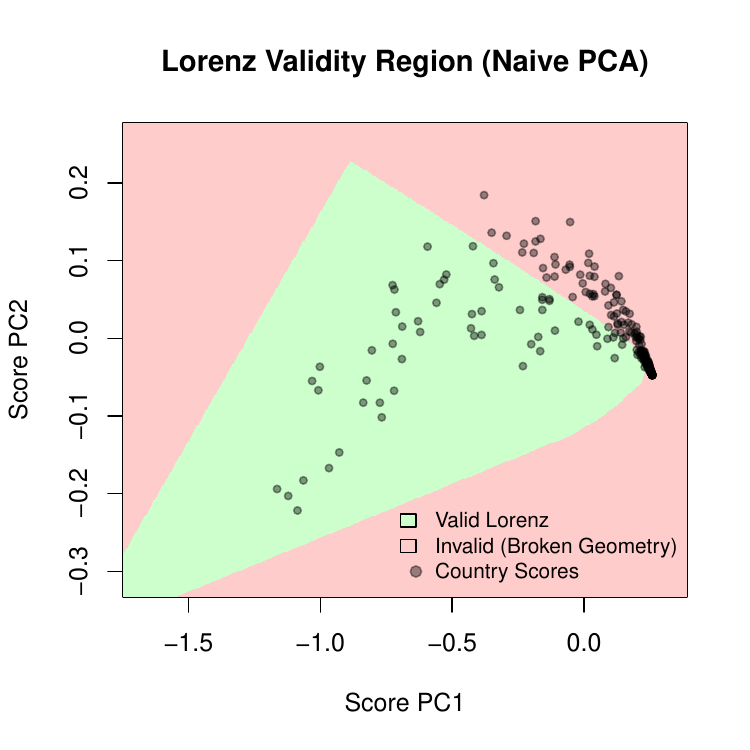}
	\caption{Upper: true \ac{LC} and truncated KL reconstruction outside $\mathcal{L}or$ for one country (positivity is violated for the reconstruction). Bottom left: modes of variation that violate the \ac{LC} constraints. Bottom right: factor plane with empirical scores and feasible region $\mathcal{F}$ (shaded); scores outside $\mathcal{F}$ are highlighted.}
	\label{fig:naive}
\end{figure}

\subsection{\ac{FPCA} after the embedding}
\label{sec:FPCA&modes}

In contrast, the \ac{FPCA} after the embedding $\psi$ is performed in the transformed space $\mathbb{R} \times L^2_{[0,1]}$, where the scalar component $b$ -- capturing the extensive margin, i.e., the share of zero-export products -- and the functional component $g$ -- capturing the intensive margin of export concentration -- are analyzed simultaneously.

Figure~\ref{fig:LC_before_after_embedding} illustrates the successive stages of the embedding. The top-left panel displays the original \acp{LC}; the top-right panel shows the corresponding transformed curves $\ell$, obtained via the perspective map $\varphi$. The bottom-left panel reports the second derivatives $\ell''$, which, interpreted as local radial inequality weights, reveal that export inequality is concentrated at the extremes of the trade value distribution: among products with no recorded trade value and among those with the highest exchange values, while intermediate trade values exhibit a comparatively egalitarian structure. Finally, the bottom-right panel displays the log-transformed quantities $\ln(-\ell'')$, which -- together with the log-odds ratios -- constitute the functional inputs for the estimation of the covariance operator and its spectral decomposition.

Figure~\ref{fig:Eigenfun_for_xi} displays the first three estimated eigenfunctions in $\mathbb{R} \times L^2_{[0,1]}$. The first estimated eigenfunction $\widehat{\nu}_{1,n}$, which explains 79.32\% of the total variance of the transformed process, does not change sign and therefore captures the variability associated with the overall level of radial inequality weight. The second and third eigenfunctions $\widehat{\nu}_{2,n}$ and $\widehat{\nu}_{3,n}$ (which explain 7.89\% and 4.59\% respectively) exhibit multiple sign changes, reflecting modes of variation that act in opposite directions across different sub-intervals of the trade value distribution, and thus capture redistributive dynamics among active exporters.

The economic interpretation of these eigenfunctions is clearer when examining the corresponding modes of variation, reported in the transformed space in Figure~\ref{fig:modesOfVar_for_g} and mapped back to the original \ac{LC} space in Figure~\ref{fig:Eigenfun_after_embedding}. Unlike the naive approach, every variation profile generated by the proposed framework is guaranteed to correspond to a mathematically valid \ac{LC}, providing a robust basis for structural decomposition and counterfactual analysis. From an economic standpoint, the leading modes of variation suggest that global export inequality is not a monolithic phenomenon, but rather a dual-structured process: the first principal component primarily reflects an extensive margin effect, opposing country profiles characterized by large shares of zero-export sectors to those with more integrated trade structures; the second and third principal components, by contrast, capture intensive margin dynamics, describing how export volumes are redistributed among active traders across different percentiles of the exchange value distribution.

\begin{figure}[ht!]
	\centering
	\includegraphics[width=.45\textwidth]{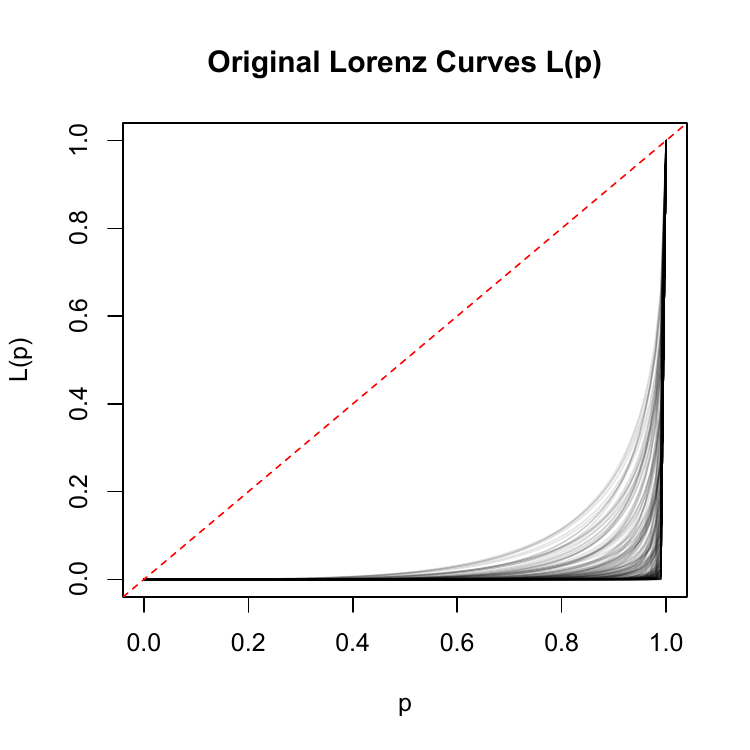}	
	\includegraphics[width=.45\textwidth]{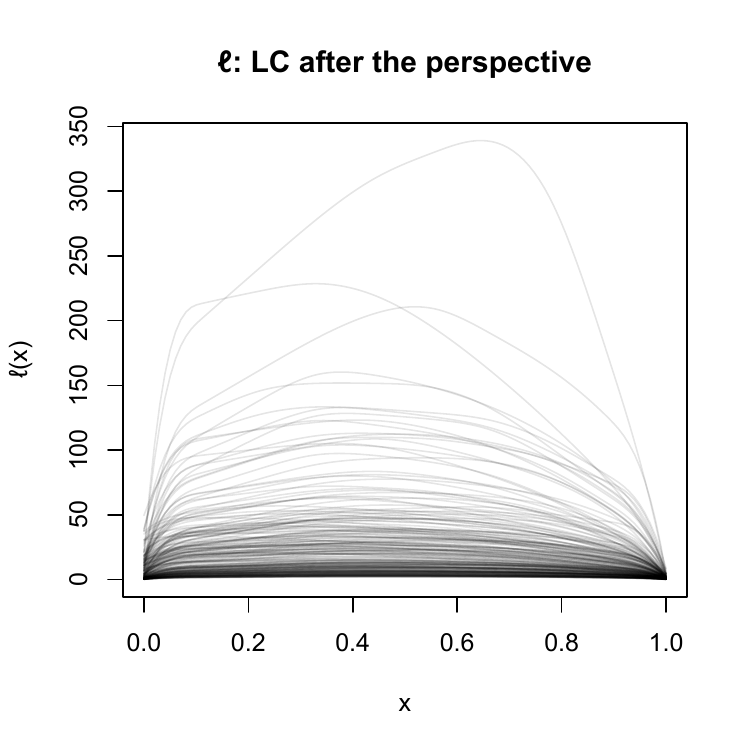}	
	\\
	\includegraphics[width=.45\textwidth]{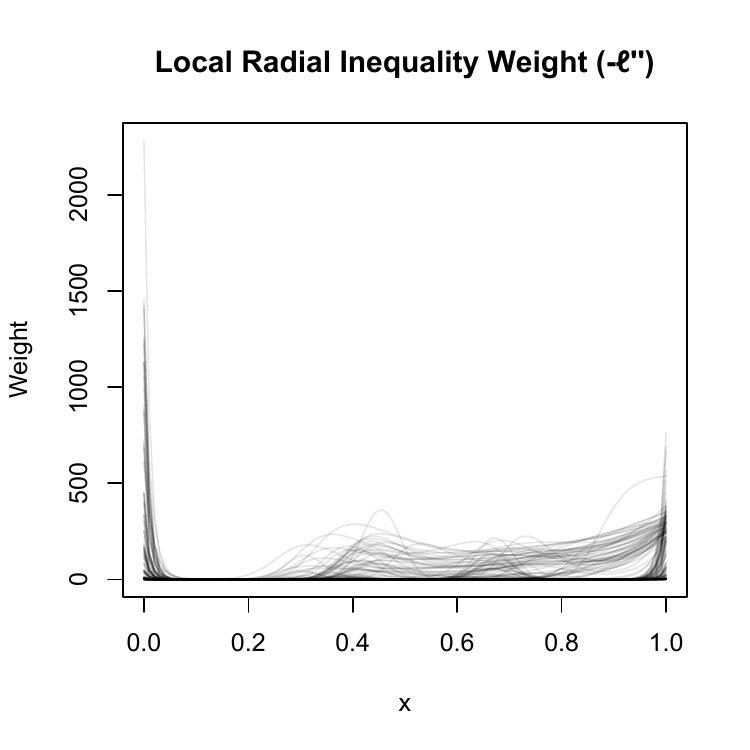}
	\includegraphics[width=.45\textwidth]{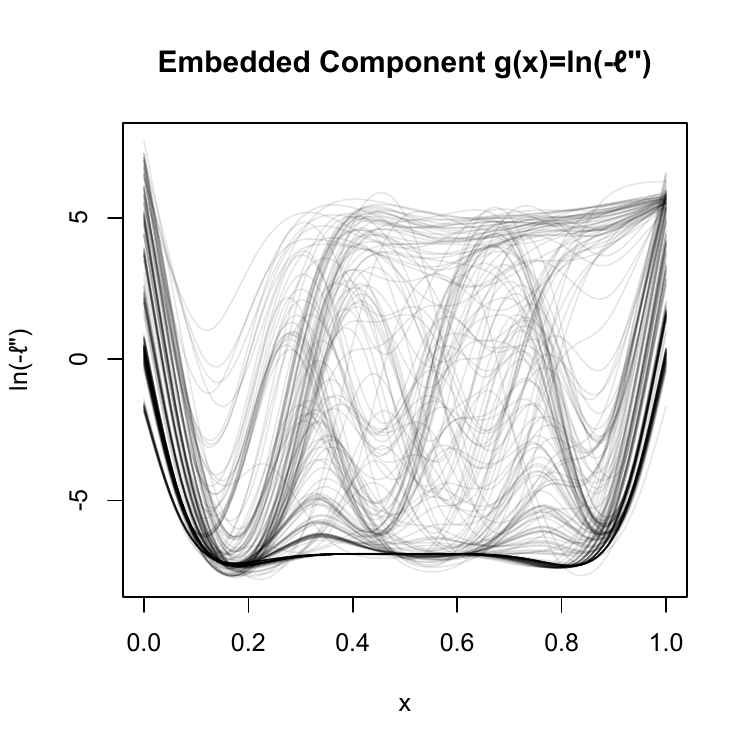}
	\caption{Upper left: original \ac{LC}. Upper right: $\ell$ functions. Bottom left: $-\ell''$ functions. Bottom right: $\ln(-\ell'')$ functions.}
	\label{fig:LC_before_after_embedding}
\end{figure}

\begin{figure}[ht!]
	\centering
	\includegraphics[width=.32\textwidth]{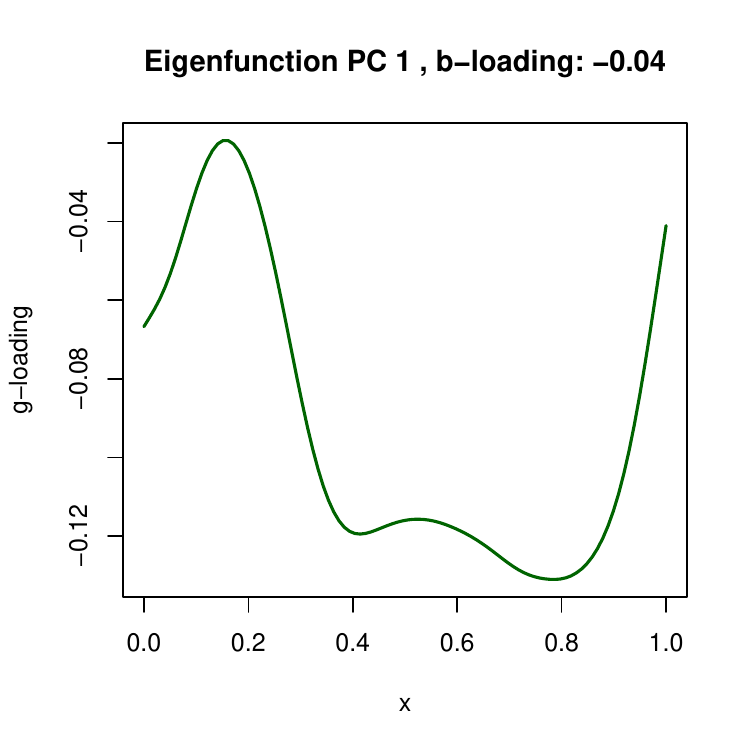}	
	\includegraphics[width=.32\textwidth]{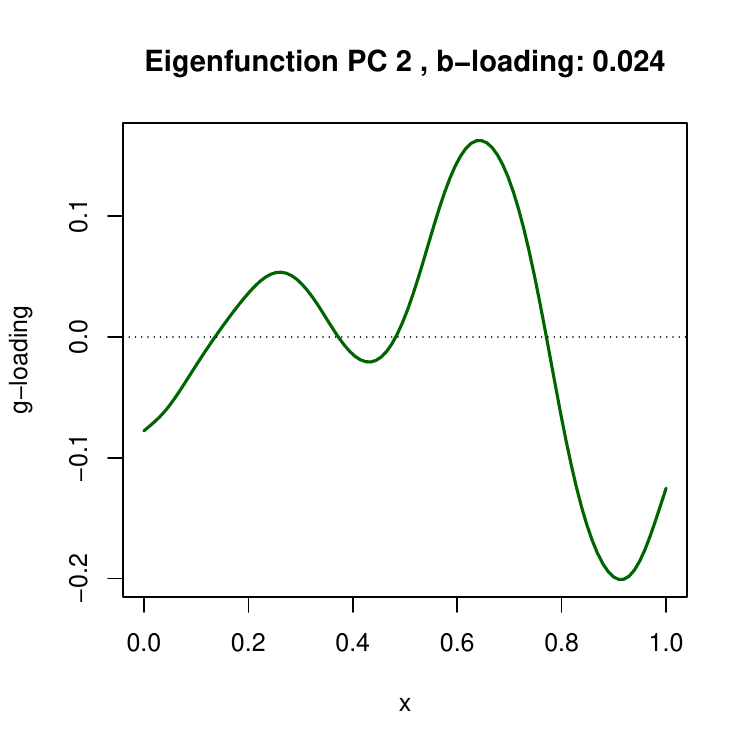}	
	\includegraphics[width=.32\textwidth]{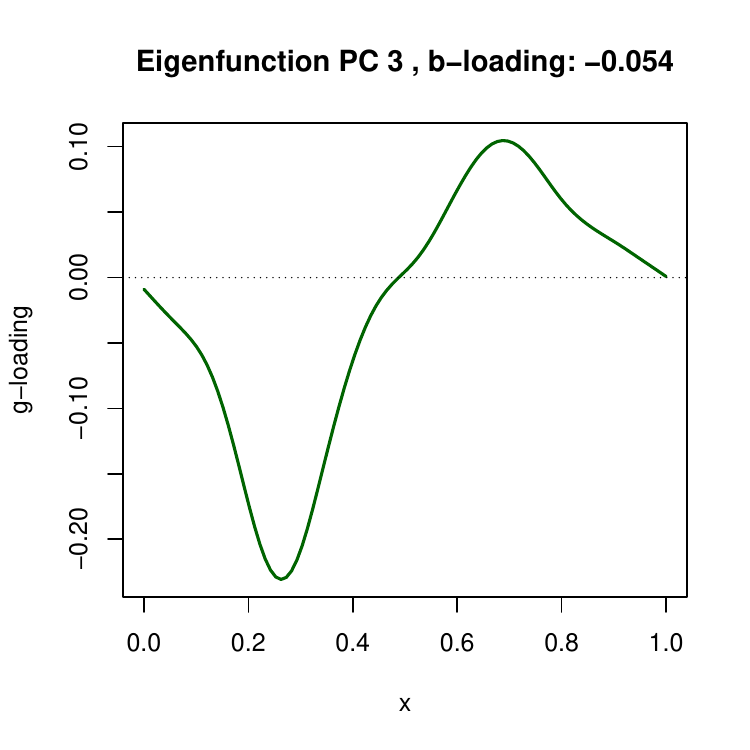}
	\caption{First three Eigenfunctions computed for $\xi$.}
	\label{fig:Eigenfun_for_xi}
\end{figure}

\begin{figure}[ht!]
	\centering
	\includegraphics[width=.32\textwidth]{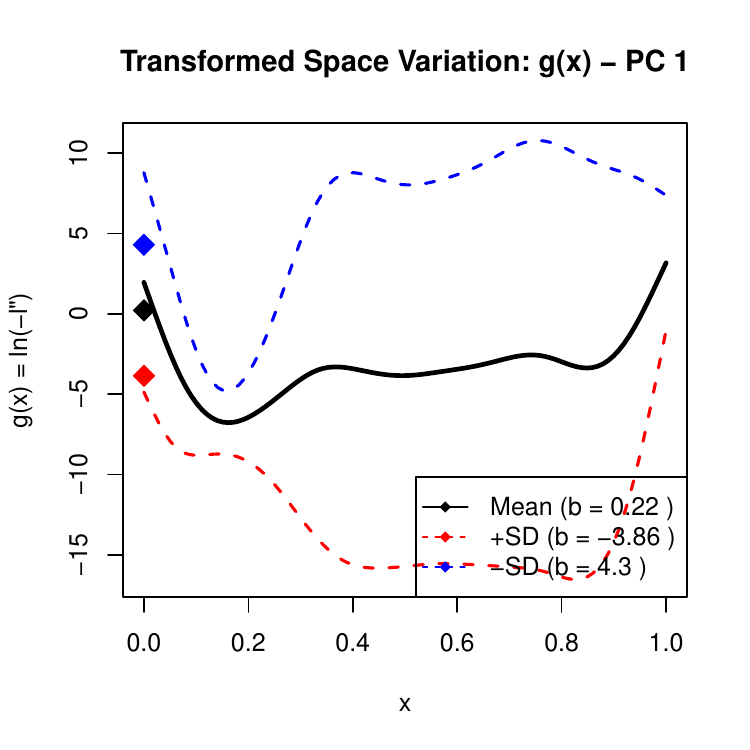}	
	\includegraphics[width=.32\textwidth]{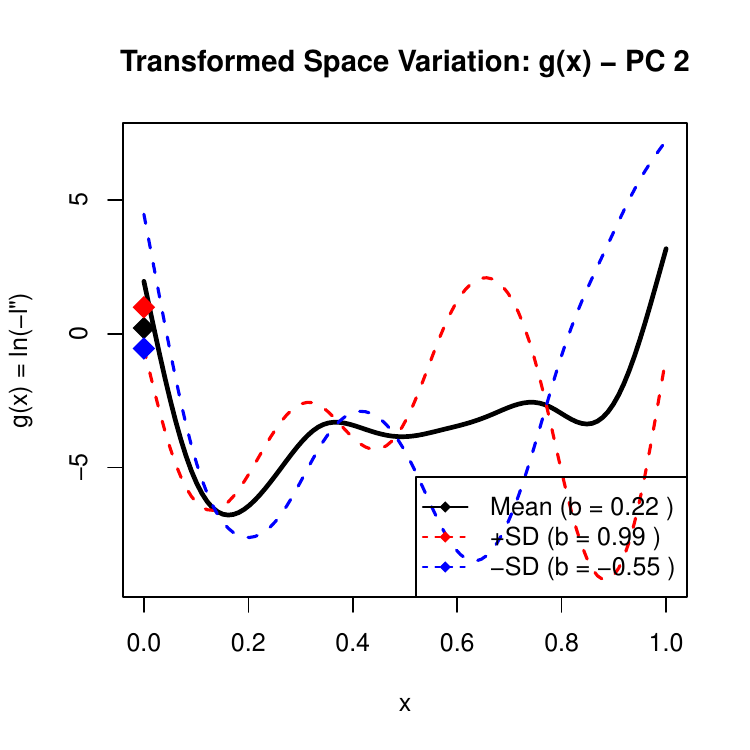}
	\includegraphics[width=.32\textwidth]{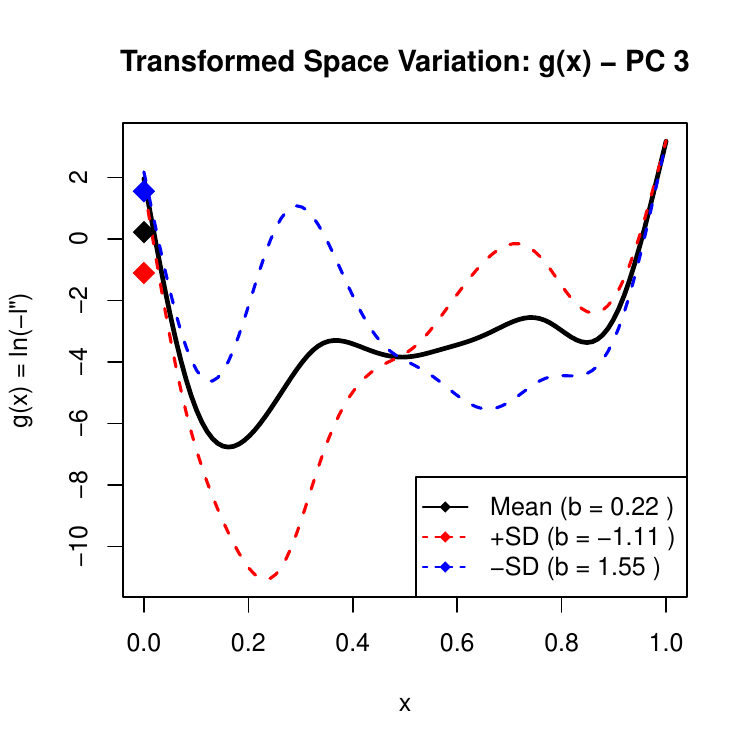}	
	\caption{First three modes of variations for $\xi$.}
	\label{fig:modesOfVar_for_g}
\end{figure}

\begin{figure}[ht!]
	\centering
	\includegraphics[width=.32\textwidth]{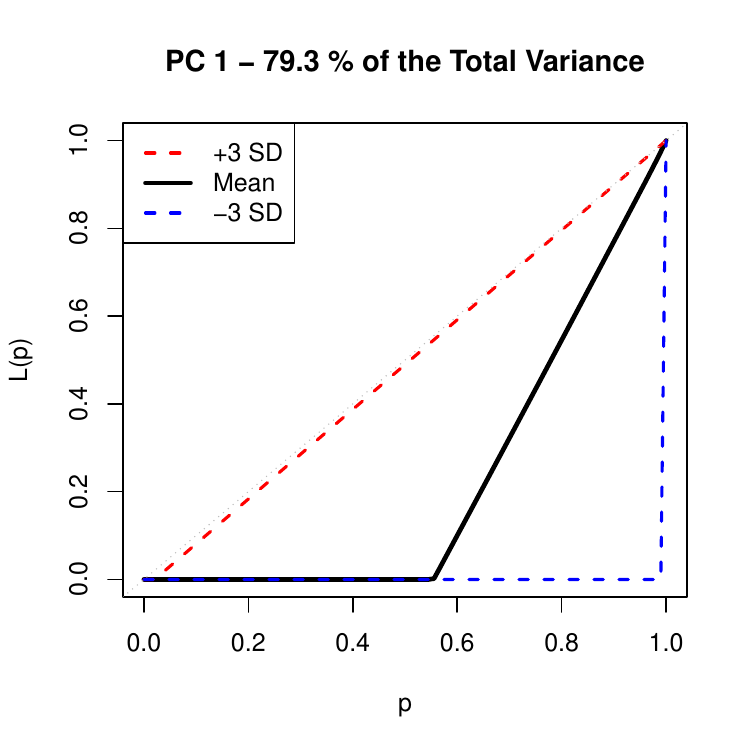}	
	\includegraphics[width=.32\textwidth]{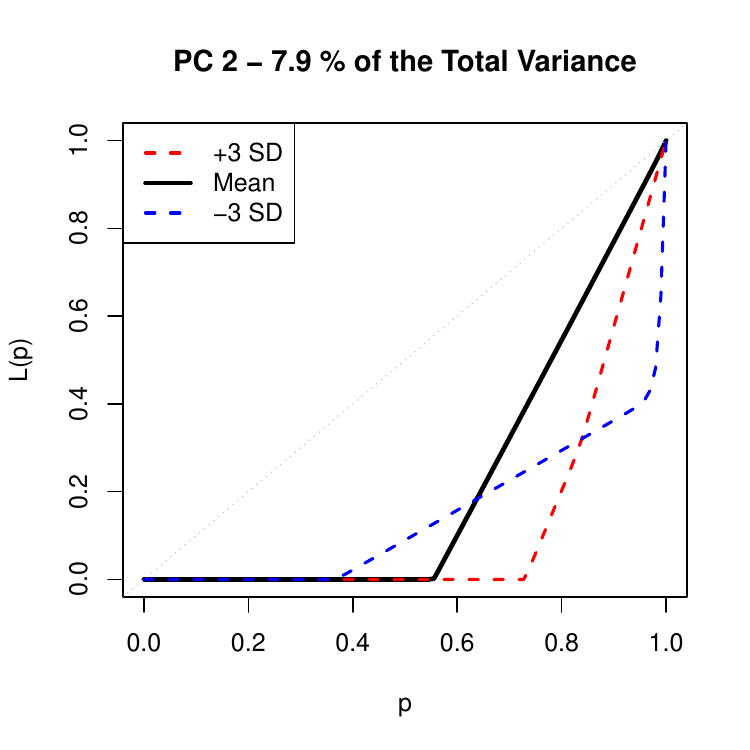}	
	\includegraphics[width=.32\textwidth]{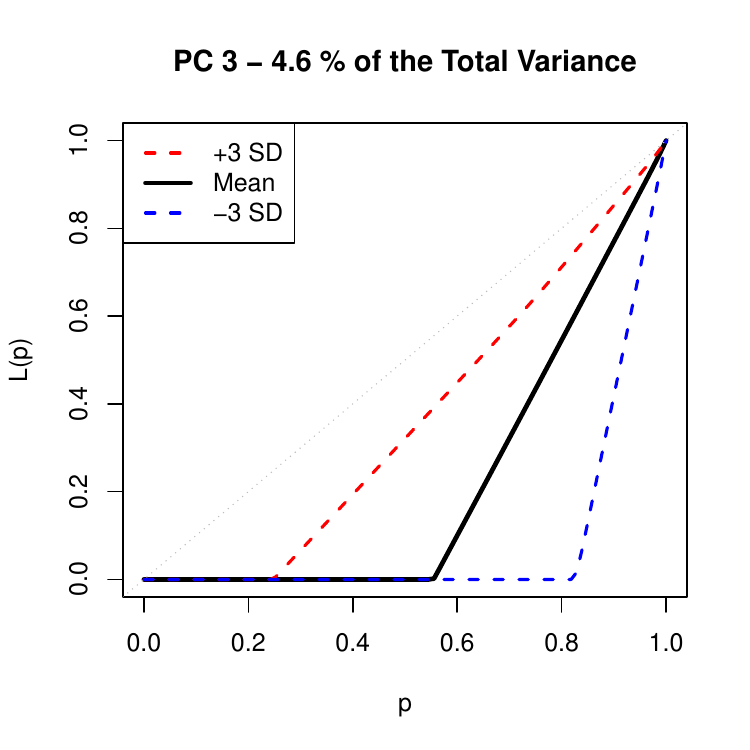}	
	\caption{First three modes of variation mapped back to the Lorenz-curve space $\mathcal{L}or$.}
	\label{fig:Eigenfun_after_embedding}
\end{figure}

\subsection{Computational Notes}
\label{sec:comp_notes}

The numerical implementation of the diffeomorphism $\psi: \mathcal{L}or \to \mathbb{R} \times C([0,1])$ requires particular care to ensure that the geometric constraints of the \ac{LC} (convexity) and its perspective transformation (concavity) are preserved throughout the functional data analysis pipeline.
While empirical \ac{LCs} $\widehat{L}_i$ are typically piecewise linear, the definition of the functional component $g = \ln(-\ell'')$ necessitates a twice-differentiable representation. To this end, the Shape-Constrained Additive Models (SCAM) are used \citep{pya2015scam}. Specifically, after the perspective transformation $\varphi$, the resulting curve $\ell$ is estimated using a B-spline basis subject to a global concavity constraint ($\ell'' \leq 0$). Unlike unconstrained smoothing splines, this approach guarantees that the argument of the logarithm in the mapping $\psi_2(L) = \ln(-\ell'')$ remains non-negative, ensuring the well-posedness of the embedding into the Hilbert space $\mathbb{H}$.
Thus the second derivative $\ell''$ is computed using a high-order finite difference approximation on the linear predictor matrix (LP-matrix) of the spline basis. Given the fitted SCAM object, let $X(x)$ denote the matrix of basis functions evaluated on an equispaced grid point $x \in [0,1]$ with small step size (e.g., $10^{-5}$). The derivative is evaluated as follows
\[
\ell''(x) \approx \frac{1}{h^2} \left( X(x+h) - 2X(x) + X(x-h) \right) \hat{\beta} , 
\]
where $\hat{\beta}$ are the estimated spline coefficients. This method is preferred to direct differentiation of the predicted values because it operates on the underlying functional basis, providing greater numerical stability. The step size $h$ is chosen to balance the truncation error of the Taylor expansion with the floating-point precision of the software.
In instances where the distribution exhibits near-linear segments in the \ac{LC} (corresponding to regions of local equality), the second derivative $\ell''$ may approach zero. To prevent numerical instability during the log-transformation, a small regularization constant $\epsilon = 10^{-3}$ is applied, that is $g = \ln(\max(-\ell'', \epsilon))$. This ensures that the functional component remains finite and integrable, as required by the $L^2_{[0,1]}$ structure.
Formally, this amounts to replacing the embedding $\psi$ with a regularized version $\psi_\varepsilon$, defined by
\[
\psi_\varepsilon(L) = \left( \ln(\ell(0)),\ \ln(\max\{-\ell'',\varepsilon\}) \right), \qquad \varepsilon>0.
\]
Since $t\mapsto\ln(\max\{t,\varepsilon\})$ is continuous and satisfies $\ln(\max\{t,\varepsilon\})\to\ln(t)$ pointwise as $\varepsilon\downarrow0$ for every $t>0$, one has $\psi_\varepsilon(L)\to\psi(L)$ in $\mathbb{H}$ as $\varepsilon\downarrow0$ whenever $-\ell''$ is bounded away from $0$ on $[0,1]$, by dominated convergence. In practice, $\varepsilon=10^{-3}$ is chosen well below the smallest curvature values observed in the estimated curves, so that $\psi_\varepsilon$ and $\psi$ coincide numerically except on the (typically negligible) near-linear segments described above; $\psi_\varepsilon$ is therefore used only as a numerically stable proxy for $\psi$, and all theoretical results of Sections~\ref{sec:diffeomorfismo}--\ref{sec:consistency_results} refer to the unregularized map $\psi$.
Finally, although Equation~\eqref{eq:ell''} establishes the exact analytical relationship between $L''(p)$ and $\ell''(x)$, computing the functional components directly from this formula is computationally unstable for two main reasons. First, estimating the first and second derivatives directly from discrete, empirical \ac{LC} coordinates severely amplifies measurement noise, leading to highly unstable and erratic profiles for $\ell''(x)$. Second, the mapping $x(p)$ is inherently country-specific. Evaluating the formula over empirical points yields irregularly spaced, non-overlapping domains across the sample. Projecting these curves back onto a common grid would require ex-post interpolation, introducing severe numerical artifacts and potentially violating local concavity ($-\ell'' > 0$).
To ensure numerical stability, the raw coordinates are first projected into the transformed space, a shape-constrained continuous spline (enforcing strict concavity) is then fitted, and its exact second derivative is evaluated on a globally uniform grid.

For the subsequent \ac{FPCA}, the scalar component $b = \ln(\ell(0))$ and the functional component $g(x)$ are concatenated into a joint operator. To respect the Hilbertian metric defined in Section~\ref{sec:consistency_results}, the covariance operator $\Sigma_\xi$ is discretized on a fine, equispaced grid of $N_{grid}$ points. The inner product in $\mathbb{H}$ is then approximated via
\[
	\langle \xi_i, \xi_j \rangle_{\mathbb{H}} \approx b_i b_j + \frac{1}{N_{grid}} \sum_{k=1}^{N_{grid}} g_i(x_k) g_j(x_k).
\]
This discretization allows for the direct application of standard \ac{FPCA} algorithms while preserving the theoretical convergence properties established by the diffeomorphism.
In the practical implementation, an inertia balancing weight $w_b$ is applied to the centered scalar component to ensure numerical consistency during the singular value decomposition. Given $n$ observations, the weight is defined as the ratio between the functional Root Mean Square (RMS) and the scalar standard deviation, that is
\[
	w_b = \frac{\sqrt{\frac{1}{n} \sum_{m=1}^{n} \frac{1}{N_{grid}} \sum_{k=1}^{N_{grid}} [g_m(x_k) - \bar{g}(x_k)]^2}}{\sigma_b}.
\]
This rescaling ensures that both structural zero-inflation (extensive margin) and the local inequality weight (intensive margin) contribute proportionately to the total variance. Without this balancing, the functional block would numerically dominate the scalar value due to the discretization over $N_{grid}$ points, artificially biasing the principal components toward local shape fluctuations and ignoring the global signal provided by the mass of zeros. 
For any fixed value $w_b>0$, the resulting Hilbert norm is equivalent to the one used in Section~\ref{sec:consistency_results}. In the empirical analysis, $w_b$ is computed once from the observed sample and then treated as a fixed normalization for the subsequent PCA. 
The resulting PCA thus respects the weighted inner product given by
\[
	\langle \psi_1, \psi_2 \rangle_{\mathbb{H}} = w_b^2 (b_1 b_2) + \frac{1}{N_{grid}} \sum_{k=1}^{N_{grid}} g_1(x_k)g_2(x_k).
\]
The reconstruction of the \ac{LC} from its embedded coordinates $(b, g) \in \mathbb{R} \times C([0,1])$ is performed by numerically inverting the transformation $\psi$. First, the perspective function $\ell$ is recovered by solving the BVP \eqref{eq:BVP_problem} via trapezoidal integration. Given the discrete grid $\mathbf{x} = \{x_1, \dots, x_{N_{grid}}\}$, the integral terms in \eqref{eq:BVP_solution} are computed cumulatively. The inverse perspective map $\varphi^{-1}$ then yields a set of pairs $(\hat{p}, \hat{L})$ which, by construction, define the \ac{LC} over its support $(\pi, 1]$, where $\pi = e^b / (1 + e^b)$. Finally, to ensure a consistent functional representation for all samples, the reconstructed curve is interpolated onto a standard equispaced grid $p \in [0,1]$, with $L(p) = 0$ for $p \leq \pi$. This numerical pipeline ensures that the estimated modes of variation and any KL-reconstructions strictly satisfy the analytical properties of the class $\mathcal{L}or$, including convexity and the presence of the structural zero segment.

\paragraph{Acknowledgments}

This publication was funded by “Bando SSHINE 2025” which has received funding from University of Piemonte Orientale.
The authors are members of the Gruppo Nazionale per l’Analisi Matematica, la Probabilità e le loro Applicazioni (GNAMPA) of the Istituto Nazionale di Alta Matematica (INdAM).
This work is based upon work from COST Action HiTEc, CA21163, supported by COST (European Cooperation in Science and Technology).

\paragraph{Declaration of generative AI and AI-assisted technologies in the writing process.}

During the preparation of this work the authors used AI tools in the writing process to improve the readability and language of the manuscript. After using these tools, the authors reviewed and edited the content as needed and take full responsibility for the content of the publication.

\appendix

\section{Appendix}

This appendix presents some geometrical properties of the perspective map $\varphi$ in Section~\ref{sec:about_varphi} that will be used in the proofs of the main results collected in Section~\ref{sec:proofs}.

\subsection{Useful properties of the perspective $\varphi$}
\label{sec:about_varphi}

The purpose of this section is to describe some simple (mostly well-known) properties of the perspective $\varphi$. Such properties will be used in the proofs collected in Appendix~\ref{sec:proofs}.

To do this, firstly consider a graphical illustration (see Figure~\ref{fig:perspectivity_visualization}) of the perspective transformation introduced in Proposition~\ref{pro:perspective}. 

\begin{figure}[ht!]
	\begin{center}
		\includegraphics[width=.9\textwidth]{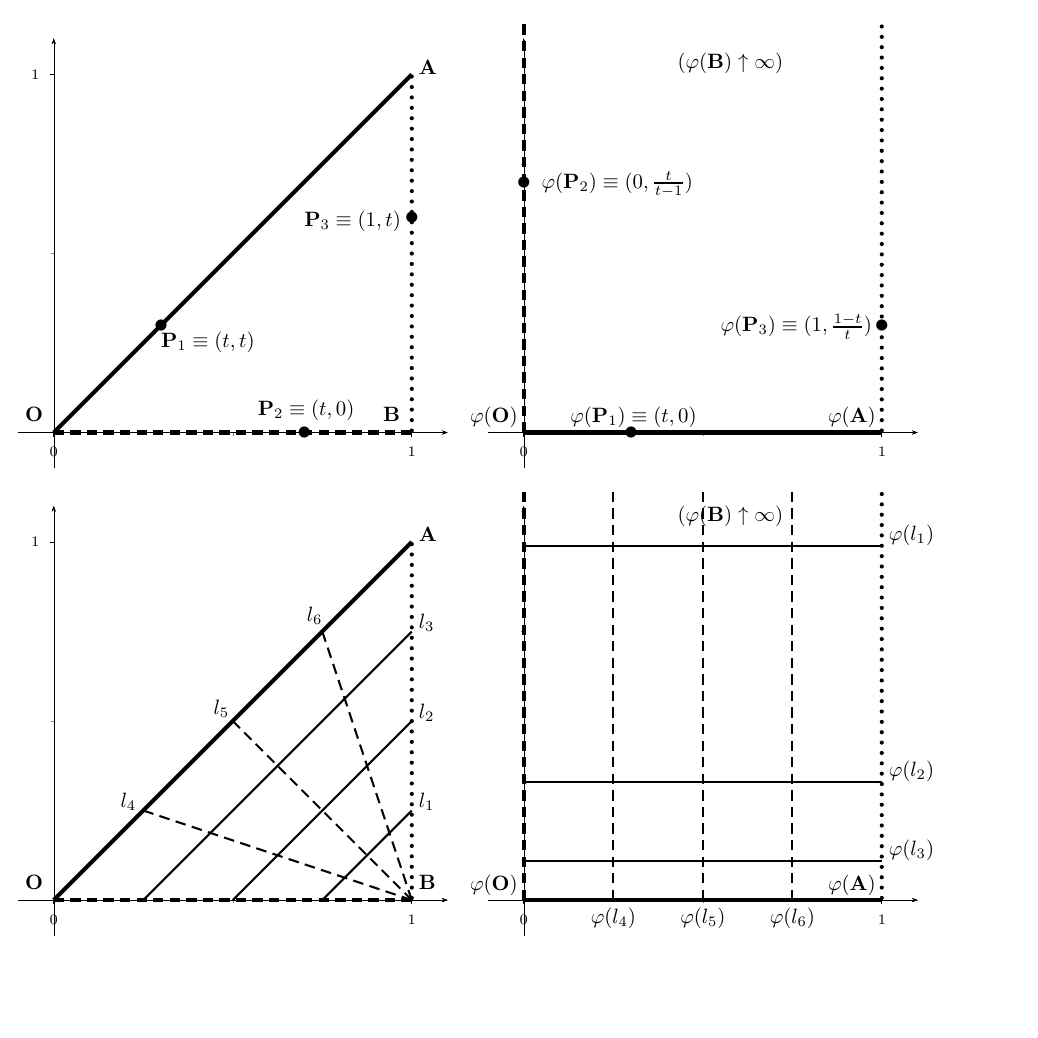}
		\caption{Perspectivity at work: right panels show how $\varphi$ maps points (first row) and lines (second row) depicted in left panels.}
		\label{fig:perspectivity_visualization}
	\end{center}
\end{figure}

Observe that the origin, $O$, is invariant,  point $A=(1,1)$ moves into $\varphi(A)=(1,0)$, and point $B=(1,0)$ moves to infinity. Moreover,  one might observe that the generic point on the main bisector $P_1=(t,t)$, with $t\in\mathbb{R}$, is mapped in $\varphi(P_1)=(t,0)$, and the diagonal of the unit square is mapped into the unit segment of the $x$--axis (solid lines). Furthermore,  the segment from $O$ to $B$ is mapped into the positive half $y$--axis (dashed lines). In particular, a generic point $P_2=(t,0)$ moves into $\varphi(P_2)=(0, \frac{t}{(1-t)})$. Finally,  the segment from $A$ to $B$ maps into the positive half line at $x=1$ (dotted lines). Indeed, the point $P_3=(1,t)$ moves into $\varphi(P_3)=(1,\frac{1-t}{t})$.  

Regarding the first objective, note that, by definition, $\widetilde{\varphi} $ moves the point $(1:0:1)\in P^2$ (i.e.\@ the point $(0,1)\in\mathbb{R}^2$) to the point at infinity $(0:1:0)$ which is the point in $P^2$ belonging to all the vertical lines\footnote{The parametric form in $P^2$ of the vertical line $x_1=a$, $\forall a\in\mathbb{R}$, is $(a:t:1)$ with $t\in\mathbb{R}$; dividing by $t$ and letting $t\to\infty$ one gets $(0:1:0)$ for any $a\in\mathbb{R}$}.
Moreover, by its definition, it is possible to see that:
\begin{enumerate}[label=\upshape{($\mathbf{\mathrm{P}_{\arabic*}}$)}, ref=\upshape{($\mathbf{\mathrm{P}_{\arabic*}}$)}]
		\item \label{proprietà-facio_rette_0}  $\varphi$ maps straight lines in straight lines;
		\item \label{proprietà-fascio_rette_1}  $\varphi$ maps the lines of equation $y=m(x-1)$ with $m<0$ into the vertical lines $x_1'=m(m-1)^{-1}$, so, as $m$ varies over $(-\infty,0]$, such lines cover the entire vertical strip $[0,1)\times \mathbb{R}$ (in general, all the straight lines passing through $(1,0)$ are mapped into vertical lines except for $y=-x-1$ that is mapped into the line at infinity);
		\item \label{proprietà-fascio_rette_2} $\varphi$ maps the line $y=0$ into $x_1'=0$;
		\item \label{proprietà-fascio_rette_3} $\varphi$ maps the line $x=1$ into $x_1'=1$;
		\item \label{proprietà-equidistribuzione} $\varphi$ maps the bisector line $y=x$ into the horizontal axis $x_2'=0$;
		\item \label{proprietà-T_e_varphi(T)} $\varphi$ maps the set $T= \{(x,y)\in \mathbb{R}^2 : 0\le x\le 1, 0\le y\le x \} \setminus \{(1,0)\}$ (i.e.\@ the closed triangle with vertexes $(0,0)$, $(1,0)$ and $(1,1)$ excluding the vertex on the boundary $(1,0)$), into the half vertical strip $\varphi(T) = [0,1]\times [0,+\infty)$;
		\item \label{proprietà-convex_into_convex} $\varphi$ maps any convex subset of $T$ in a convex subset of $\varphi(T)$ (this is due to the fact that any convex set can be expressed as the intersection of halfspaces that are defined via straight lines with the latter being mapped into straight lines thanks to \ref{proprietà-facio_rette_0}).
		\item \label{proprietà-varphi_bicontinua} $\varphi$ (the inverse $\varphi^{-1}$ resp.) is a diffeomorphism if restricted over $T$ ($\varphi(T)$ resp.) because it is a rational function, see~\eqref{eq:varphi_change_coordinates} (\eqref{eq:varphi_inverse} resp.), defined over a domain without singularities.
	\end{enumerate}
	Most of the above properties are illustrated in Figure~\ref{fig:perspectivity_visualization}.

\subsection{Proofs}
\label{sec:proofs}

\begin{proof}[Proof of Proposition~\ref{pro:perspective}.]
	Let ${P}^2$ be the projective plane whose points have homogeneous coordinates $(x_1:x_2:x_0)$, i.e.\@ the equivalence class of all points in $\mathbb{R}^3$ that are equal to $(x_1,x_2,x_0)$ up to a scaling factor (see e.g.\@ \cite{ber2009}). 
	With this setting, any point $(x_1,x_2)$ in the plane $\mathbb{R}^2$ can be seen as a point in $P^2$ with coordinates $(x_1:x_2:1) \equiv (tx_1:tx_2:t)$, for any not null real number $t$, whereas the set of point in $P^2$ satisfying $x_0 = 0$ is the line at infinity. 
	Moreover, the following equation
	\begin{equation*}
		\widetilde{\varphi} : \qquad 
		\left(
		\begin{array}{c}
			x_1'\\
			x_2'\\
			x_0'
		\end{array}
		\right)
		=
		\left(
		\begin{array}{ccc}
			a & b & 0\\
			c & d & 0\\
			h & k & 1
		\end{array}
		\right)
		\left(
		\begin{array}{c}
			x_1\\
			x_2\\
			x_0
		\end{array}
		\right)
	\end{equation*}
	represents the generic form of a perspective on the projective plane that, mapping $P^2$ to itself, can be uniquely identified by three noncollinear points and their associated images \citep{ber2009}. 
	In this view, the perspective introduced in Section~\ref{sec:about_varphi} is obtained imposing the following conditions:
	\begin{equation} \label{eq:perspective_conditions}
		\left\{
		\begin{array}{ll}
			\widetilde{\varphi}(0:0:1) = (0:0:1) \qquad & \textnormal{the origin is mapped into itself},\\
			\widetilde{\varphi}(1:1:1) = (1:0:1) & \textnormal{the point }(1,1)\textnormal{ is mapped into }(1,0),\\
			\widetilde{\varphi}(1:0:1) = (0:1:0) & \textnormal{the point }(1,0)\textnormal{ is moved to the infinity point}
			\\
			& \textnormal{belonging to all vertical lines},
		\end{array}
		\right.
	\end{equation}
	from which one gets
	\[
	\widetilde{\varphi} : \qquad 
	\left(\begin{array}{c}
		x_1'\\
		x_2'\\
		x_0'
	\end{array}\right)
	=
	\left(\begin{array}{c}
		x_2\\
		x_1-x_2\\
		x_2-x_1+x_0
	\end{array}\right),
	\]
	and, writing the perspective as a map of $\mathbb{R}^2$ to itself (i.e.\@ dividing the first two coordinates by $x_0'$),
	\begin{equation}
		\label{eq:varphi_change_coordinates}
		\varphi : \qquad 
		\left(\begin{array}{c}
			x_1'\\
			x_2'
		\end{array}\right)
		=
		\left(\begin{array}{c}
			x_2/(x_2-x_1+1)\\
			(x_1-x_2)/(x_2-x_1+1)
		\end{array}\right),
	\end{equation}
	that straightforwardly leads to \eqref{Eq:persp} with a trivial change in the variable names.
	
	Regarding the inverse $\varphi^{-1}$, reverting the conditions in \eqref{eq:perspective_conditions} one gets
	\[
	\widetilde{\varphi}^{-1} : \qquad 
	\left(\begin{array}{c}
		x_1\\
		x_2\\
		x_0
	\end{array}\right)
	=
	\left(\begin{array}{c}
		x_1'+x_2'\\
		x_1'\\
		x_2'+x_0'
	\end{array}\right),
	\]
	that, written in terms of $\mathbb{R}^2$ coordinates, becomes
	\begin{equation}\label{eq:varphi_inverse}
		\varphi^{-1} : \qquad 
		\left(\begin{array}{c}
			x_1\\
			x_2
		\end{array}\right)
		=
		\left(\begin{array}{c}
			(x_1' + x_2') / (x_2' + 1)\\
			x_1'/(x_2' + 1)
		\end{array}\right),
	\end{equation}
	which concludes the proof. 
\end{proof}

\begin{proof}[Proof of Proposition~\ref{pro:definition_L-tilde}.]
	The proof is divided into three steps. First, the function $\varphi$ transforms $L$ into an almost everywhere $C^2$ parametric curve whose graph is contained in $\varphi(T)=[0,1]\times[0,+\infty)$. Second, this parametric curve defines %
	a function $\ell$, defined on $[0,1]$ and taking values in $[0,+\infty)$, whose graph coincides with that of $\varphi(p, L(p))$ for $p\in [\pi, 1]$. Third, the regularity properties of $\ell$ are deduced.
	\\
	\underline{FIRST STEP.} Note that by definition of $\varphi$ the image of $L$ is the parametric curve in $\mathbb{R}^2$ given by
	\[
	\varphi(p, L(p)) = \left(\frac{L(p)}{L(p)-p+1} ,\frac{p-L(p)}{L(p)-p+1}\right),\qquad p\in[0,1].
	\tag{\ref{Eq:persp}}
	\]
	Now, since the graph of $L$ is a subset of the simplex $T$, its image through $\varphi$ is a subset of $\varphi(T)$ (see \ref{proprietà-T_e_varphi(T)}). Moreover, by \ref{proprietà-varphi_bicontinua}, since $L\in\mathcal{L}or$ and $\pi=\inf \mathrm{supp}(L)$, it is straightforward to see that, the parametric curve \eqref{Eq:persp} is twice differentiable for all $p\in (\pi,1]$ and right differentiable for $p=\pi$.
	
	\noindent \underline{SECOND STEP.} The idea is to show that the graph of  \eqref{Eq:persp} intersect the vertical line $x_1'=0$ at least at one point while intersecting any vertical line $x_1'=k$ with $k\in(0,1]$ at only one point. 
	Regarding  $k=0$ it is sufficient to observe that, for $p \in [0,\pi]$, $L(p)=0$ and hence the graph of  \eqref{Eq:persp} coincides with the segment $\{0\}\times [0, \pi(1-\pi)^{-1}]$ that lies on $x_1'=0$.
	For $k\in(0,1)$, because of \ref{proprietà-fascio_rette_1}, the uniqueness of the intersection with the graph of   \eqref{Eq:persp} is a consequence of the uniqueness of the intersection of $L$ with the line $y=m(x-1)$ with $m<0$. 
	Finally, for $k=1$, note that $x=1$ is invariant for $\varphi$ (see \ref{proprietà-fascio_rette_3}) and intersects $L$ only at the point $(1,1)$, thus the intersection of $x_1'=1$ with the graph of   \eqref{Eq:persp} is unique at the point $(1,0)$.
	In other words, the parametric curve \eqref{Eq:persp} allows the function $\varphi(L)=\ell: [0,1] \to [0,+\infty)$ to be defined and its graph coincides with that of $\varphi(p, L(p))$ restricted to $p\in [\pi, 1]$. As a consequence, $\ell(0) = \pi(1-\pi)^{-1}$ and $\ell(1)=0$.
	\\
	\underline{THIRD STEP.} As a consequence of the regularity properties seen at the first step and thanks to \ref{Lor-prop:derivable_at_q} and \ref{Lor-prop:derivable_at_1} (see Definition~\ref{def:Lorenz_curves}) $\ell$ is twice differentiable with continuous and bounded derivatives over $[0,1]$. About the concavity of the function, note that since $L$ is a convex function, the locus of points between $L$ and the bisector line i.e.\@ $\Lambda =\{(x,y) : x\in [0,1] \text{ and } L(x) \le y\le x \}$ is a convex set. Thanks to \ref{proprietà-convex_into_convex}, $\varphi(\Lambda)$ is still convex and equals the locus of points between $\ell$ and the horizontal axis implying that the function $\ell$ is concave over $[0,1]$.
	Explicitly, from Equation~\eqref{eq:ell''}, as $p\downarrow\pi$, using $L\left(\pi\right)=0$ and the positive finite right derivatives $L'\left(\pi^+\right),L''\left(\pi^+\right)$ granted by~\ref{Lor-prop:derivable_at_q},
	\[
	\ell'\left(x\left(\pi^+\right)\right) = \frac{1+\pi-L'\left(\pi^+\right)}{L'\left(\pi^+\right)\left(1-\pi\right)}, \qquad
	\ell''\left(x\left(\pi^+\right)\right) = \frac{-L''\left(\pi^+\right)}{\left(L'\left(\pi^+\right)\right)^3},
	\]
	where $x\left(\pi^+\right)=0$; both limits are finite since $L'\left(\pi^+\right),L''\left(\pi^+\right)\in\left(0,\infty\right)$. Similarly, as $p\uparrow1$, using $L\left(1\right)=1$ and the bounded second left derivative $L''\left(1^-\right)$ granted by~\ref{Lor-prop:derivable_at_1} (which in particular implies $L'\left(1^-\right)<\infty$),
	\[
	\ell'\left(x\left(1^-\right)\right) = 2-L'\left(1^-\right), \qquad
	\ell''\left(x\left(1^-\right)\right) = -L''\left(1^-\right),
	\]
	where $x\left(1^-\right)=1$; both are finite. These explicit limits confirm that $\ell'$ and $\ell''$ extend continuously to $x=0$ and $x=1$ respectively, consistently with the regularity already established at the first step.
\end{proof}

\begin{proof}[Proof of Proposition~\ref{pro:BVP_solution}.]
	The general solution of the BVP problem \eqref{eq:BVP_problem} is
	\begin{equation*}
		u(p) =c_{1}+c_{2}p+\int_{0}^{p}\int_{0}^{z}f( t) dtdz.
	\end{equation*}
	The boundary conditions lead to $c_{1}=a$ and $c_{2}=-a-\int_{0}^{1}\int_{0}^{z}f( t) dtdz$. By integration by parts, the latter can be rewritten as $c_{2}=-a-\int_{0}^{1}(1-z) f( z) dz$. Hence one has
	\begin{equation*}
		u(p) =a-p\left[ a+\int_{0}^{1}(1-z) f( z) dz \right] +\int_{0}^{p}\int_{0}^{z}f ( t) dtdz
	\end{equation*}
	which, integrated by parts, leads to
	\begin{equation*}
		u(p) =a-p\left[ a+\int_{0}^{1}(1-z) f( z) dz \right] +\int_{0}^{p} ( p-z ) f( z ) dz
	\end{equation*}
	and, by straightforward calculation, to
	\begin{equation*}
		BVP(a,f) = u(p) = a(1-p) + ( p-1) \int_{0}^{p}zf( z) dz+p\int_{p}^{1}( z-1) f( z) dz. \tag{\ref{eq:BVP_solution}}
	\end{equation*}
	
	Finally, because $f$ is continuous on $[0,1]$ and hence integrable, the BVP's solution is well-posed and unique, which means that it depends continuously on the input data $a$ and $f$. %
\end{proof}

The following Lemma is preparatory for the proof of Proposition~\ref{pro:psi_bi-continuous_bijection}.

\begin{lemma}\label{lem:surjectivity}
	For every $\left(b,g\right)\in\mathbb{R}\times C\left(\left[0,1\right]\right)$, the function $\varphi^{-1}\left[BVP\left(e^b,-e^g\right)\right]$ belongs to $\mathcal{L}or$.
\end{lemma}
\begin{proof}
	Set $a=e^b>0$, $f=-e^g$, and $u=BVP\left(a,f\right)$, so that $u\in C^2\left(\left[0,1\right]\right)$, $u''=f<0$ everywhere, $u\left(0\right)=a$, $u\left(1\right)=0$ (Proposition~\ref{pro:BVP_solution}); in particular $u$ is strictly concave.
	
	\smallskip
	\noindent\underline{Non-negativity.} By concavity, writing $p=\left(1-p\right)\cdot0+p\cdot1$ for $p\in\left[0,1\right]$,
	\[
	u\left(p\right) \ge \left(1-p\right)u\left(0\right)+p\,u\left(1\right) = a\left(1-p\right) \ge 0.
	\]
	
	\smallskip
	\noindent\underline{The change of variables $p\left(x\right)$.} Define $p:\left[0,1\right]\to\mathbb{R}$ by $p\left(x\right)=\left(x+u\left(x\right)\right)/\left(u\left(x\right)+1\right)$, so that
	\[
	p'\left(x\right) = \frac{u\left(x\right)+1+u'\left(x\right)\left(1-x\right)}{\left(u\left(x\right)+1\right)^2}.
	\]
	Since $u$ is concave, the tangent line at $x$ lies above $u$, so $u\left(1\right)\le u\left(x\right)+u'\left(x\right)\left(1-x\right)$, i.e.\@ $u\left(x\right)+u'\left(x\right)\left(1-x\right)\ge0$; hence $p'\left(x\right)\ge1/\left(u\left(x\right)+1\right)^2>0$ for every $x\in\left[0,1\right]$. Thus $p$ is a strictly increasing $C^1$-diffeomorphism from $\left[0,1\right]$ onto $\left[\pi,1\right]$, where $\pi:=p\left(0\right)=a/\left(a+1\right)\in\left(0,1\right)$ and $p\left(1\right)=1$.
	
	\smallskip
	\noindent\underline{Definition of $L$.} Set $L\left(p\right)=0$ for $p\in\left[0,\pi\right]$, and $L\left(p\left(x\right)\right)=x/\left(u\left(x\right)+1\right)$ for $x\in\left(0,1\right]$. Since $u\ge0$, one has $L\left(p\left(x\right)\right)=x/\left(u\left(x\right)+1\right)\le x\le1$, and $L\left(p\left(x\right)\right)>0$ for $x>0$; together with $L\left(1\right)=1/\left(u\left(1\right)+1\right)=1$, this gives~\ref{Lor-prop:1} and~\ref{Lor-prop:support}.
	
	\smallskip
	\noindent\underline{Monotonicity.} By the inverse function theorem, for $x\in\left(0,1\right)$,
	\[
	L'\left(p\left(x\right)\right) = \frac{dL/dx}{dp/dx} = \frac{u\left(x\right)+1-xu'\left(x\right)}{u\left(x\right)+1+u'\left(x\right)\left(1-x\right)}.
	\]
	Again by concavity, the tangent line at $x$ lies above $u$ at $0$: $u\left(0\right)\le u\left(x\right)-xu'\left(x\right)$, i.e.\@ $u\left(x\right)+1-xu'\left(x\right)\ge a+1>0$; combined with the positivity of the denominator shown above, $L'\left(p\left(x\right)\right)>0$ for every $x\in\left(0,1\right)$, so $L$ is strictly increasing on $\left(\pi,1\right)$, giving the increasing part of~\ref{Lor-prop:conv_incr}.
	
	\smallskip
	\noindent\underline{Convexity.} The hypograph $\Lambda' = \left\{\left(x,y\right):x\in\left[0,1\right],\,0\le y\le u\left(x\right)\right\}$ is convex, since $u$ is concave. As $\varphi^{-1}$ maps straight lines to straight lines \ref{proprietà-facio_rette_0}, it maps convex subsets of $\varphi\left(T\right)$ into convex subsets of $T$, by the same argument as~\ref{proprietà-convex_into_convex}. Since $\varphi^{-1}\left(x,0\right)=\left(x,x\right)$ and $\varphi^{-1}\left(x,u\left(x\right)\right)=\left(p\left(x\right),L\left(p\left(x\right)\right)\right)$, the image $\varphi^{-1}\left(\Lambda'\right)$ is the region $\left\{\left(p,z\right):L\left(p\right)\le z\le p\right\}$, which is therefore convex; as in the proof of Proposition~\ref{pro:definition_L-tilde}, this is equivalent to $L$ being convex on $\left[0,1\right]$, completing~\ref{Lor-prop:conv_incr}.
		
	\smallskip
	\noindent\underline{Regularity and boundary derivatives.} Since $\varphi^{-1}$ restricted to $\varphi\left(T\right)$ is a diffeomorphism onto $T$ \ref{proprietà-varphi_bicontinua} and $u\in C^2\left(\left[0,1\right]\right)$, $L$ is twice differentiable on $\left(0,1\right]\setminus\left\{\pi\right\}$, giving the differentiability part of~\ref{Lor-prop:derivable_at_q}.
	Since $u$ is concave, its tangent line at $x$ lies above $u$ both at $0$ and at $1$, giving
	\[
	M\left(x\right):=u\left(x\right)+1-xu'\left(x\right)\ge a+1>0, \qquad N\left(x\right):=u\left(x\right)+1+u'\left(x\right)\left(1-x\right)\ge1>0
	\]
	for every $x\in\left[0,1\right]$. Since also $D\left(x\right):=u\left(x\right)+1\ge1>0$ and $u''\left(x\right)=f\left(x\right)=-e^{g\left(x\right)}<0$ for every $x\in\left[0,1\right]$ (as $g$ is real-valued), differentiating $L'\left(p\left(x\right)\right)=M\left(x\right)/N\left(x\right)$ with respect to $x$ and dividing by $p'\left(x\right)=N\left(x\right)/D\left(x\right)^2$ gives
	\[
	L''\left(p\left(x\right)\right) = -u''\left(x\right)\,\frac{\left[xN\left(x\right)+\left(1-x\right)M\left(x\right)\right]D\left(x\right)^2}{N\left(x\right)^3},
	\]
	which is finite and strictly positive for every $x\in\left[0,1\right]$, i.e.\@ $L''\left(p\right)>0$ for every $p\in\left(\pi,1\right)$, with positive one-sided limits at $p=\pi$ and $p=1$. This establishes the strict positivity of $L''$ on $\left(\pi,1\right)$ required in~\ref{Lor-prop:derivable_at_q}, together with the strict positivity at $p=1$ required in~\ref{Lor-prop:derivable_at_1}.
	For later reference, the explicit boundary values follow from the same formula: at $x=0$, since $D\left(0\right)=M\left(0\right)=a+1$ and $N\left(0\right)\ge1$,
	\[
	L'\left(\pi^+\right) = \frac{M\left(0\right)}{N\left(0\right)} = \frac{a+1}{N\left(0\right)} \in\left(0,\infty\right), \qquad
	L''\left(\pi^+\right) = e^{g\left(0\right)}\,\frac{\left(a+1\right)^3}{N\left(0\right)^3} \in\left(0,\infty\right);
	\]
	at $x=1$, since $D\left(1\right)=N\left(1\right)=1$,
	\[
	L''\left(1^-\right) = e^{g\left(1\right)} \in\left(0,\infty\right).
	\]
	Together with the previous steps, this completes~\ref{Lor-prop:derivable_at_q} and~\ref{Lor-prop:derivable_at_1}, and shows $\varphi^{-1}\left[BVP\left(e^b,-e^g\right)\right]\in\mathcal{L}or$.
\end{proof}

\begin{proof}[Proof of Proposition~\ref{pro:psi_bi-continuous_bijection}.]
	Firstly prove that $\psi$ is a bijection. It is enough to show that $\psi^{-1}\{\psi (L)\}=L$ and $\psi \{\psi ^{-1}(b,g)\}=(b,g)$ for any $L\in \mathcal{L}or$ and $(b,g)\in \mathbb{R}\times C([0,1])$ (see \cite{mcc88}).
	\\
	Concerning the first condition $\psi^{-1}\{\psi (L)\}=L$, by definition of $\psi$ and $\psi^{-1}$, for each $L\in\mathcal{L}or$, it holds
	\begin{align*}
		\psi^{-1}\{\psi (L)\} & = \psi^{-1}\left\{ (\ln(\ell(0)),\ln(-\ell''))\right\} \\
		& =\varphi^{-1}\left\{ BVP\left[ \exp(\ln(\ell(0))),-\exp(\ln(-\ell'') ) \right]\right\}
		\\
		& = \varphi^{-1}\left\{ BVP( \ell(0) , \ell'')\right\}
		= \varphi^{-1}\left\{ \ell \right\} = L.
	\end{align*}
	\\
	By Lemma~\ref{lem:surjectivity}, $\varphi^{-1}\left[BVP(e^b,-e^g)\right]\in\mathcal{L}or$ for every $(b,g)$, so $\psi^{-1}$ is well defined as a map into $\mathcal{L}or$; it remains to verify $\psi\{\psi^{-1}(b,g)\}=(b,g)$, $\forall (b,g)\in \mathbb{R}\times C([0,1])$. To do this, note that
	\begin{align*}
		\psi \{\psi^{-1}(b,g)\} & = \psi \left\{\varphi^{-1}\left[ BVP( e^b, -e^g ) \right]\right\} 
		\\
		& = ( \psi_1( \varphi^{-1}\left[ BVP( e^b, -e^g ) \right] ) , 
		\psi_2(\varphi^{-1}\left[ BVP( e^b, -e^g ) \right] ) )
	\end{align*}
	where $\psi_1$, $\psi_2$ are defined by \eqref{eq:bijection_psi}. Thus it is sufficient to prove that 
	\begin{align}
		\psi_1( \varphi^{-1}\left[ BVP( e^b, -e^g ) \right] ) & =b, \notag
		\\
		\psi_2(\varphi^{-1}\left[ BVP( e^b, -e^g ) \right] ) & = g. \label{eq:proof_psi2}
	\end{align}
	For the first, note that $\psi_1( \varphi^{-1}\left[ BVP( e^b, -e^g ) \right] )= \ln(\pi(1-\pi)^{-1})$ where $\pi = \inf \mathrm{supp}(\varphi^{-1}\left[ BVP( e^b, -e^g ) \right] )$. Since $BVP( e^b, -e^g ) $ is a function passing through $(x_1',x_2')=(0,e^b)$ (and with second derivative $-e^g$) then the support of the function $\varphi^{-1}\left[ BVP( e^b, -e^g ) \right]$ is the interval $[\pi,1]$ where
	$\pi = \frac{x_1'+x_2'}{x_2'+1} = \frac{e^b}{e^b+1}$ and consequently $\ln(\pi(1-\pi)^{-1}) = b$.
	For the second identity in \eqref{eq:proof_psi2}, note that by definition of $\psi_2$
	\begin{align*}
		\psi_2(\varphi^{-1}\left[ BVP( e^b, -e^g ) \right] ) & =\ln( -D^2\varphi( \varphi^{-1}\left[ BVP( e^b, -e^g ) \right] ) ) 
		\\
		& = \ln( -D^2\left[ BVP( e^b, -e^g ) \right] )  = \ln( e^g ) = g.
	\end{align*}
	
	It remains to justify that $\psi$ and $\psi^{-1}$ are continuous with respect to the stated metrics, and to specify the smooth structure on $\mathcal{L}or$.
	\\
	For $\psi^{-1}$: by Proposition~\ref{pro:BVP_solution}, $BVP(a,f)$ depends continuously on $(a,f)\in\mathbb{R}\times C([0,1])$ in sup norm; the map $(b,g)\mapsto(e^b,-e^g)$ is continuous $\mathbb{R}\times C([0,1])\to\mathbb{R}\times C([0,1])$; and $\varphi^{-1}$ restricted to $\varphi(T)$ is continuous onto $T$ by~\ref{proprietà-varphi_bicontinua}. Composing these three continuous maps, and using Lemma~\ref{lem:surjectivity} to know the composition lands in $\mathcal{L}or$, shows that $\psi^{-1}$ is continuous from $\mathbb{R}\times C([0,1])$ to $\left(\mathcal{L}or,\|\cdot\|_\infty\right)$.
	\\
	For $\psi$: given $L\in\mathcal{L}or$, $\varphi$ restricted to $T$ is continuous onto $\varphi(T)$ by~\ref{proprietà-varphi_bicontinua}, giving continuity of $L\mapsto\ell$; since $\ell''(x)<0$ for every $x\in[0,1]$ by~\ref{Lor-prop:derivable_at_q}--\ref{Lor-prop:derivable_at_1} (see the remark after Definition~\ref{def:Lorenz_curves}), the map $\ell\mapsto\ln(-\ell'')$ is well defined and continuous on $[0,1]$, and $\ell(0)>0$ makes $\ell(0)\mapsto\ln(\ell(0))$ continuous as well. Hence $\psi$ is continuous from $\left(\mathcal{L}or,\|\cdot\|_\infty\right)$ to $\mathbb{R}\times C([0,1])$.
	\\
	Finally, equipping $\mathcal{L}or$ with the smooth structure pulled back through $\psi$ from $\mathbb{R}\times C([0,1])$ -- i.e.\@ declaring $h:\mathcal{L}or\to\mathbb{R}$ smooth whenever $h\circ\psi^{-1}$ is smooth on $\mathbb{R}\times C([0,1])$ -- $\psi$ becomes, by construction, a diffeomorphism from $\mathcal{L}or$ onto $\mathbb{R}\times C([0,1])$, consistently with the metric topology established above.
\end{proof}

The remaining proofs are based on the following technical and classical result.
\begin{lemma}\label{lem:boundedness_fourth_moment}
	Under assumptions \ref{hp:momento_4_q} and \ref{hp:momento_4_L}, the fourth moment of $\Vert \xi \Vert_\mathbb{H}$ is bounded, i.e.\@ $\mathbb{E}\left[ \left\Vert \xi \right\Vert_\mathbb{H}^{4} \right]<\infty$. Moreover, the fourth moment of $\xi$ in the Banach space $\mathbb{B}=\mathbb{R}\times C\left(\left[0,1\right]\right)$ is also bounded, i.e.\@ $\mathbb{E}\left[ \left\Vert \xi \right\Vert_\mathbb{B}^{4} \right]<\infty$.
\end{lemma}
\begin{proof}[Proof of Lemma~\ref{lem:boundedness_fourth_moment}.]
	First, consider the fourth moment of $\Vert \xi \Vert_\mathbb{H}$
	\begin{align*}
		\mathbb{E}\left[ \left\Vert \xi \right\Vert_\mathbb{H}^{4} \right]  &= \mathbb{E}\left[ \left\Vert \psi(L) \right\Vert_\mathbb{H}^{4} \right] = \mathbb{E}\left[ \left\Vert (\ln(\ell(0)) , \ln(-\ell'')) \right\Vert_\mathbb{H}^{4} \right]  
		\\
		& =  \mathbb{E}\left[ \left(\ln^2 (\ell(0)) + \left\Vert \ln(-\ell'')\right\Vert_{L^2_{[0,1]}}^2\right)^{2} \right]
		\\
		& \le 2 \mathbb{E}\left[ \ln^4 (\ell(0)) + \left\Vert \ln(-\ell'') \right\Vert_{L^2_{[0,1]}}^4  \right]
	\end{align*}
	where the last inequality is a consequence of $\left(a+b\right)^2 \le 2a^2+2b^2$, for any $a,b \in \mathbb{R}$. Similarly,
	\begin{align*}
		\mathbb{E}\left[ \left\Vert \xi \right\Vert_\mathbb{B}^{4} \right]  
		& = \mathbb{E}\left[ \left( \left\vert \ln(\ell(0)) \right\vert + \left\Vert \ln(-\ell'') \right\Vert_\infty \right)^4 \right] 
		\le 8 \mathbb{E}\left[ \ln^4 (\ell(0)) + \left\Vert \ln(-\ell'') \right\Vert_\infty^4  \right],
	\end{align*}
	using $\left(a+b\right)^4 \le 8\left(a^4+b^4\right)$ for any $a,b\ge0$. Thanks to Assumption~\ref{hp:momento_4_q}, $\mathbb{E}\left[ \ln^4 (\ell(0)) \right]$ is finite, and it remains to control $\mathbb{E}\left[ \left\Vert \ln(-\ell'') \right\Vert_{L^2_{[0,1]}}^4 \right]$ and $\mathbb{E}\left[ \left\Vert \ln(-\ell'') \right\Vert_\infty^4 \right]$.
	By definition,
	\[
	\left\Vert \ln(-\ell'') \right\Vert_\infty^4 = \left( \sup_{x\in\left[0,1\right]} \left\vert \ln\left(-\ell''\left(x\right)\right) \right\vert \right)^4 = \sup_{x\in\left[0,1\right]} \left\vert \ln\left(-\ell''\left(x\right)\right) \right\vert^4,
	\]
	so that Assumption~\ref{hp:momento_4_L} directly gives $\mathbb{E}\left[ \left\Vert \ln(-\ell'') \right\Vert_\infty^4 \right] < \infty$, which proves the second moment bound.
	
	For the first, since $\left[0,1\right]$ has unit length, Jensen's inequality gives
	\[
	\left\Vert \ln(-\ell'') \right\Vert_{L^2_{[0,1]}}^4 = \left( \int_0^1 \ln^2\left(-\ell''\left(x\right)\right) dx \right)^2 \le \int_0^1 \ln^4\left(-\ell''\left(x\right)\right) dx \le \sup_{x\in\left[0,1\right]} \left\vert \ln\left(-\ell''\left(x\right)\right) \right\vert^4,
	\]
	so that, again by Assumption~\ref{hp:momento_4_L}, $\mathbb{E}\left[ \left\Vert \ln(-\ell'') \right\Vert_{L^2_{[0,1]}}^4 \right] < \infty$. Combining the two bounds concludes the proof.
\end{proof}

\begin{proof}[Proof of Proposition~\ref{pro:eigenfunction_continuous}.] 
	Thanks to Lemma~\ref{lem:boundedness_fourth_moment} (which requires \ref{hp:momento_4_q} and \ref{hp:momento_4_L}), the fourth moment of $\Vert \xi \Vert_\mathbb{H}$ and the fourth moment of $\Vert \xi \Vert_\mathbb{B}$ are both finite. By Jensen's inequality,
	\[
	\left(\mathbb{E}\left[ \Vert \xi \Vert_\mathbb{H}^{2} \right] \right)^2 \le \mathbb{E}\left[ \Vert \xi \Vert_\mathbb{H}^{4} \right] <\infty,
	\qquad
	\left(\mathbb{E}\left[ \Vert \xi \Vert_\mathbb{B}^{2} \right] \right)^2 \le \mathbb{E}\left[ \Vert \xi \Vert_\mathbb{B}^{4} \right] <\infty,
	\]
	and, consequently,
	\[
	\mathbb{E}\left[ \left\vert \xi^{(1)} \right\vert^{2} \right]<\infty ,
	\qquad \qquad
	\mathbb{E}\left[ \left\Vert \xi^{(2)} \right\Vert_{L^2_{[0,1]}}^{2} \right]<\infty,
	\qquad \qquad
	\mathbb{E}\left[ \left\Vert \xi^{(2)} \right\Vert_\infty^{2} \right]<\infty.
	\]
	By Proposition~\ref{pro:definition_L-tilde}, $\ell\in C^2\left(\left[0,1\right]\right)$ almost surely, so each sample path of $\xi^{(2)} = \ln\left(-\ell''\right)$ is continuous on $\left[0,1\right]$ and, in particular,
	\[
	\xi^{(2)}\left(x+h\right) - \xi^{(2)}\left(x\right) \longrightarrow 0 \qquad \text{a.s.\@ as } h\to0,
	\]
	for every $x\in\left[0,1\right]$. Moreover,
	\[
	\left\vert \xi^{(2)}\left(x+h\right) - \xi^{(2)}\left(x\right) \right\vert^2 \le \left( 2\left\Vert \xi^{(2)} \right\Vert_\infty \right)^2 \qquad \text{a.s.\@ for every } h,
	\]
	then the dominated convergence Theorem yields
	\[
	\mathbb{E}\left[\left(\xi^{(2)}\left(x+h\right) - \xi^{(2)}\left(x\right)\right)^2\right] \longrightarrow 0 \qquad \text{as } h\to0,
	\]
	i.e.\@ $\xi^{(2)}$ is continuous in quadratic mean on $\left[0,1\right]$.
	\\%
	To show that $Cov(\xi^{(1)}, \xi^{(2)}(x))$ is continuous with respect to $x$ note that the linearity property of the covariance and the Cauchy--Schwarz inequality lead to
	\begin{align*}
		\left\vert Cov(\xi^{(1)}, \xi^{(2)}(x+h)) - Cov(\xi^{(1)}, \xi^{(2)}(x)) \right\vert = 
		\left\vert Cov(\xi^{(1)}, \xi^{(2)}(x+h) - \xi^{(2)}(x))\right\vert \le
		\\
		\le 
		\sigma_{\xi^{(1)}} 
		\left( 
		\mathbb{E}\left[\left(\xi^{(2)}(x+h) - \xi^{(2)}(x)\right)^2\right]  - 
		\left( \mathbb{E}\left[ \xi^{(2)}(x+h)\right] - \mathbb{E}\left[\xi^{(2)}(x)\right]\right)^2
		\right)^{1/2}.
	\end{align*}
	For $h\to 0$, the left--hand side tends to zero because $\xi^{(2)}$ is continuous in quadratic mean.
	Regarding the covariance function $c(x_1,x_2)=Cov(\xi^{(2)}(x_1), \xi^{(2)}(x_2))$, the triangle inequality gives
	\begin{align*}
		\left\vert c(x_1+h_1\right. & \left. \!\!, x_2+h_2) - c(x_1,x_2) \right\vert = 
		\\ 
		& = 
		\left\vert c(x_1+h_1,x_2+h_2) - c(x_1+h_1,x_2) + c(x_1+h_1,x_2) - c(x_1,x_2) \right\vert \le 
		\\
		& \le \left\vert Cov(\xi^{(2)}(x_1+h_1), \xi^{(2)}(x_2+h_2) - \xi^{(2)}(x_2)) \right\vert +
		\\
		& \qquad\qquad + \left\vert Cov(\xi^{(2)}(x_1+h_1) - \xi^{(2)}(x_1),  \xi^{(2)}(x_2)) \right\vert.
	\end{align*}
	Using arguments analogous to those previously employed for the cross-covariance, and exploiting the mean-square continuity of $\xi^{(2)}$, it follows that the covariance function $c(x_1,x_2)$ is continuous on $[0,1]\times [0,1]$.
	\\%
	Let $\nu = [\nu^{(1)}, \nu^{(2)}]^{\intercal}\in \mathbb{H}$ be an eigenvector of the covariance operator $\Sigma_{\xi}$ associated with eigenvalue $\alpha$. Then it satisfies the equation $\Sigma_{\xi}\nu = \alpha \nu$ that is
	\[
	\alpha \left[
	\begin{array}{c}
		\nu^{(1)}
		\\
		\nu^{(2)}
	\end{array}
	\right] =
	\left[
	\begin{array}{cc}
		\sigma^2_{\xi^{(1)}}	& \Sigma_{\xi^{(1)}, \xi^{(2)}}
		\\ 
		\Sigma_{\xi^{(2)}, \xi^{(1)}} & \Sigma_{\xi^{(2)}}
	\end{array}
	\right]
	\left[
	\begin{array}{c}
		\nu^{(1)}
		\\
		\nu^{(2)}
	\end{array}
	\right] = 
	\left[
	\begin{array}{c}
		\nu^{(1)} \sigma^2_{\xi^{(1)}} + \Sigma_{\xi^{(1)}, \xi^{(2)}} \nu^{(2)}
		\\
		\nu^{(1)} \Sigma_{\xi^{(2)}, \xi^{(1)}} + \Sigma_{\xi^{(2)}} \nu^{(2)}
	\end{array}
	\right].
	\]
	Thus, 
	\begin{equation}\label{eq:second_component_eigenvalue}
		\nu^{(2)} = \frac{1}{\alpha}\left(\nu^{(1)} \Sigma_{\xi^{(2)}, \xi^{(1)}} + \Sigma_{\xi^{(2)}} \nu^{(2)} \right)
	\end{equation}
	is continuous on $[0,1]$ thanks to results established above and, hence, $\nu\in \mathbb{B}$.
\end{proof}

\begin{proof}[Proof of Proposition~\ref{pro:cons-MoV-full-observed-LCs}.]
 First of all, note that since $M_{j}^{k}=\psi ^{-1}(m_{j}^{k})$ and $M_{j,n}^{k}=\psi ^{-1}(m_{j,n}^{k})$, the continuity of $\psi ^{-1}$ from $\mathbb{B}$ to $\left(\mathcal{L}or,\Vert\cdot\Vert_\infty\right)$ (see Proposition~\ref{pro:psi_bi-continuous_bijection}) implies that the a.s.\@ convergence of $m_{j,n}^{k}$ to $m_{j}^{k}$ in $\mathbb{B}$ is sufficient to guarantee that $\Vert M_{j,n}^{k}-M_{j}^{k}\Vert_\infty\to0$ a.s..
The convergence in $\mathbb{B}$ is established in two steps: the following triangle inequality controls the $\mathbb{H}$-norm of the difference, and the uniform convergence of the functional component is proved afterwards. Consider then
\begin{align}
	\left\Vert m_{j,n}^k-m_{j}^k \right\Vert_\mathbb{H} & = \left\Vert \mu_{n}\pm k\nu _{j,n} (\alpha _{j,n})^{1/2} - 
	\left( \mu \pm k\nu _{j}(\alpha _{j})^{1/2} \right) \right\Vert_\mathbb{H}  \notag \\
	&\leq \left\Vert \mu_{n}-\mu \right\Vert_\mathbb{H} + k  \left( (\alpha _{j,n})^{1/2} \Vert \nu _{j,n}-\nu _{j} \Vert_\mathbb{H} \right. \notag\\
	& \qquad\qquad\qquad\qquad  \left. + \Vert \nu _{j} \Vert_\mathbb{H} \times
	\left\vert (\alpha _{j,n})^{1/2}- (\alpha _{j})^{1/2} \right\vert \right),
	\label{eq:tri_ineq_mode_of_variation_Lor_tilde_primo}
\end{align}
where, to avoid identification problems, $\nu _{j,n}$ is taken such that $\langle \nu _{j,n},\nu _{j}\rangle_\mathbb{H} \ge 0$.
Concerning the first summand in \eqref{eq:tri_ineq_mode_of_variation_Lor_tilde_primo}, note that, by Lemma~\ref{lem:boundedness_fourth_moment} and Jensen's inequality, $\mathbb{E}\left[\Vert\xi\Vert_\mathbb{B}\right]<\infty$; since $\mu_{n}$ is the empirical mean of i.i.d.\@ variables with bounded first moment in the Banach space $\mathbb{B}=\mathbb{R}\times C([0,1])$, the Strong Law of Large Numbers (see e.g.\@ \citealp[Chapter~7]{led:tal13}) guarantees that $\Vert \mu_{n}-\mu\Vert_\mathbb{B} \rightarrow 0 $ a.s.. This implies that $\Vert \mu_{n}-\mu\Vert_\mathbb{H} \rightarrow 0 $ and that the functional component converges uniformly on $[0,1]$.
Consider now the remaining terms of \eqref{eq:tri_ineq_mode_of_variation_Lor_tilde_primo}, that is
\begin{equation*}
	k \left( (\alpha _{j,n})^{1/2} \, \Vert \nu _{j,n}-\nu _{j} \Vert_\mathbb{H} + \Vert
	\nu _{j} \Vert_\mathbb{H}\ \left\vert (\alpha _{j,n})^{1/2}-(\alpha _{j})^{1/2}
	\right\vert \right).
\end{equation*}
Using the fact that $\Vert \nu _{j}\Vert_\mathbb{H} =1$, $\alpha _{1,n}\geq \alpha_{j,n}$ and $\mathbb{E}\left[ \Vert \xi \Vert_\mathbb{H}^{2} \right] <\infty $ (see Lemma~\ref{lem:boundedness_fourth_moment}) together with standard consistency results for the empirical covariance operator and its eigenelements $ \{\alpha _{j,n},\nu _{j,n}\} $ (see \citealp[Chapter~4]{bos00}), the following results hold, for a suitable positive constant $C$,
\begin{align}
	& \left\vert (\alpha _{j,n})^{1/2}-(\alpha _{j})^{1/2} \right\vert \leq 
	C \Vert \Sigma _{\xi ,n}-\Sigma _{\xi } \Vert _{\infty } \to 0, \qquad a.s.,
	\label{eq:eigenvalues_dominated_by_covariance}
	\\
	& 
	\Vert \nu _{j,n}-\nu _{j} \Vert_\mathbb{H} \leq \frac{2\sqrt{2}}{\delta_j} \Vert \Sigma _{\xi ,n}-\Sigma _{\xi } \Vert _{\infty }\to 0, \qquad a.s.,
	\label{eq:eigenvectors_dominated_by_covariance}
\end{align}
where $\delta_j$ is defined as above (with the convention $\delta_1=\alpha_1-\alpha_2$ for $j=1$).
To prove the uniform convergence of the functional component of $m_{j,n}^k$, it is sufficient to show that the second component of the eigenfunctions, denoted by $\nu_{j,n}^{(2)}$, converges uniformly to $\nu_{j}^{(2)}$ on $[0,1]$. 
Analogously to what was done to get Equation~\eqref{eq:second_component_eigenvalue}, the functional component of the eigenvalue equation is:
\begin{equation*}
	\nu_{j,n}^{(2)}(x_1) = \frac{1}{\alpha_{j,n}} \left( \nu_{j,n}^{(1)} c_{2,1,n}(x_1) + \int_0^1 c_n(x_1,x_2) \nu_{j,n}^{(2)}(x_2) dx_2 \right), \qquad x_1 \in [0,1],
\end{equation*}
where 
\begin{align*}
c_{2,1,n}(x_1) & = \frac1n\sum_{i=1}^n\left(\xi_i^{(2)}(x_1)-\mu_n^{(2)}(x_1)\right)\left(\xi_i^{(1)}-\mu_n^{(1)}\right),
\\
c_n(x_1,x_2) & = \frac1n\sum_{i=1}^n\left(\xi_i^{(2)}(x_1)-\mu_n^{(2)}(x_1)\right)\left(\xi_i^{(2)}(x_2)-\mu_n^{(2)}(x_2)\right)
\end{align*}
are the empirical counterparts of the cross-covariance function $\Sigma_{\xi^{(2)},\xi^{(1)}}$ and of the covariance function $c$ of $\xi^{(2)}$, respectively (see Equation~\eqref{eq:block_theoretical_cov}).
Thanks to Proposition~\ref{pro:eigenfunction_continuous}, the theoretical covariance function $c(x_1,x_2)$ is continuous on $[0,1]^2$.
Since $c_n$ is centred at the empirical mean $\mu_n^{(2)}$, it is not itself an average of i.i.d.\@ terms; however, it differs from such an average only by a term vanishing with $\mu_n-\mu$. Indeed,
\[
c_n=\frac1n\sum_{i=1}^n\left(\xi_i^{(2)}-\mu^{(2)}\right)\otimes\left(\xi_i^{(2)}-\mu^{(2)}\right)-\left(\mu_n^{(2)}-\mu^{(2)}\right)\otimes\left(\mu_n^{(2)}-\mu^{(2)}\right),
\]
where $\left(f\otimes h\right)(x_1,x_2) := f(x_1)h(x_2)$, and the second term has sup norm bounded by $\Vert\mu_n-\mu\Vert_\mathbb{B}^2\to0$ a.s.. It is therefore enough to study the first term, which is the empirical mean of i.i.d.\@ elements of $C([0,1]^2)$. Since
\[
\left\Vert \left(\xi_i^{(2)}-\mu^{(2)}\right)\otimes\left(\xi_i^{(2)}-\mu^{(2)}\right) \right\Vert_{C([0,1]^2)}
= \left\Vert \xi_i^{(2)}-\mu^{(2)} \right\Vert_\infty^2 \le \left\Vert \xi_i-\mu \right\Vert_\mathbb{B}^2,
\]
Lemma~\ref{lem:boundedness_fourth_moment} together with Jensen's inequality gives
\[
\mathbb{E}\left[ \left\Vert \left(\xi_i^{(2)}-\mu^{(2)}\right)\otimes\left(\xi_i^{(2)}-\mu^{(2)}\right) \right\Vert_{C([0,1]^2)} \right] \le \mathbb{E}\left[\left\Vert \xi_i-\mu \right\Vert_\mathbb{B}^2\right] <\infty,
\]
so that the Strong Law of Large Numbers in $C([0,1]^2)$ applies and, combined with the previous remark, ensures that $\Vert c_n - c \Vert_\infty \to 0$ a.s..
The same argument, applied to $c_{2,1,n}$ and to the i.i.d.\@ elements $\left(\xi_i^{(2)}-\mu^{(2)}\right)\left(\xi_i^{(1)}-\mu^{(1)}\right)$ of $C([0,1])$, gives $\Vert c_{2,1,n}-\Sigma_{\xi^{(2)},\xi^{(1)}}\Vert_\infty\to0$ a.s..

By adding and subtracting the term $\alpha_{j,n}^{-1} \int_0^1 c_n(x_1,x_2) \nu_{j}^{(2)}(x_2) dx_2$, the following bound for the difference is obtained:
\begin{align*}
	\vert \nu_{j,n}^{(2)}(x_1) & - \nu_{j}^{(2)}(x_1) \vert \leq \frac{1}{\alpha_{j,n}} \left\vert \int_0^1 c_n(x_1,x_2) \left( \nu_{j,n}^{(2)}(x_2) - \nu_{j}^{(2)}(x_2) \right) dx_2 \right\vert \\
	& \quad + \left\vert \frac{1}{\alpha_{j,n}} \int_0^1 c_n(x_1,x_2) \nu_{j}^{(2)}(x_2) dx_2 - \frac{1}{\alpha_{j}} \int_0^1 c(x_1,x_2) \nu_{j}^{(2)}(x_2) dx_2 \right\vert.
\end{align*}
Applying the Cauchy--Schwarz inequality to the first integral gives
\begin{equation*}
	\left\vert \int_0^1 c_n(x_1,x_2) \left( \nu_{j,n}^{(2)}(x_2) - \nu_{j}^{(2)}(x_2) \right) dx_2 \right\vert 
	\leq \!
	\left( \int_0^1 c_n(x_1,x_2)^2 dx_2 \right)^{1/2} 
	\!\!\!\!\! \Vert \nu_{j,n}^{(2)} - \nu_{j}^{(2)} \Vert_{L^2_{[0,1]}}.
\end{equation*}
The first factor is uniformly bounded a.s.\@ due to the uniform convergence of $c_n$, while $\Vert \nu_{j,n}^{(2)} - \nu_{j}^{(2)} \Vert_{L^2_{[0,1]}} \to 0$ as a consequence of the convergence in $\mathbb{H}$ established above. Thus, the first term vanishes uniformly in $x_1$. Similarly, the second term converges to zero uniformly on $[0,1]$ due to the a.s.\@ convergence of $\alpha_{j,n} \to \alpha_j$ and the uniform convergence of the covariance functions.
Consequently, $\Vert m_{j,n}^k - m_{j}^k \Vert_\mathbb{B} \to 0$ a.s..

Regarding $L^2(\mathbb{P})$ convergence, since $\Vert m_{j,n}^k - m_{j}^k \Vert_{\mathbb{B}}\to 0$ a.s., it is enough to prove that the family $\{\Vert m_{j,n}^k - m_{j}^k \Vert_{\mathbb{B}}^2\}_{n}$ is uniformly integrable (see \citealp[Theorem 4, p.~188]{shi2016}).
To this end, note that the empirical eigenpair satisfies $\alpha_{j,n}\nu_{j,n}=\Sigma_{\xi,n}\nu_{j,n}$ with $\Vert\nu_{j,n}\Vert_\mathbb{H}=1$, that is, componentwise,
\[
\alpha_{j,n}\,\nu_{j,n}^{(1)}=\frac1n\sum_{i=1}^n\big(\xi_i^{(1)}-\mu_n^{(1)}\big)\langle \xi_i-\mu_n,\nu_{j,n}\rangle_\mathbb{H},
\qquad
\alpha_{j,n}\,\nu_{j,n}^{(2)}(x)=\frac1n\sum_{i=1}^n\big(\xi_i^{(2)}(x)-\mu_n^{(2)}(x)\big)\langle \xi_i-\mu_n,\nu_{j,n}\rangle_\mathbb{H},
\]
for $x\in[0,1]$, while $n^{-1}\sum_{i=1}^n\langle \xi_i-\mu_n,\nu_{j,n}\rangle_\mathbb{H}^2=\langle\Sigma_{\xi,n}\nu_{j,n},\nu_{j,n}\rangle_\mathbb{H}=\alpha_{j,n}$.
By the Cauchy--Schwarz inequality, the factor $\sqrt{\alpha_{j,n}}$ cancels out and, whenever $\alpha_{j,n}>0$,
\begin{equation}\label{eq:CS_eigen_bound}
	\sqrt{\alpha_{j,n}}\,\big\vert\nu_{j,n}^{(1)}\big\vert+\sqrt{\alpha_{j,n}}\,\big\Vert\nu_{j,n}^{(2)}\big\Vert_\infty
	\le 2\Big(\frac1n\sum_{i=1}^n\Vert\xi_i-\mu_n\Vert_\mathbb{B}^2\Big)^{1/2}
	\le 4\,T_n^{1/2},
	\qquad
	T_n:=\frac1n\sum_{i=1}^n\Vert\xi_i-\mu\Vert_\mathbb{B}^2,
\end{equation}
where the last inequality follows from $\Vert\xi_i-\mu_n\Vert_\mathbb{B}^2\le 2\Vert\xi_i-\mu\Vert_\mathbb{B}^2+2\Vert\mu_n-\mu\Vert_\mathbb{B}^2$ and $\Vert\mu_n-\mu\Vert_\mathbb{B}^2\le T_n$ (Jensen's inequality). Hence
\[
\Vert m_{j,n}^k - m_{j}^k \Vert_{\mathbb{B}}^2
\le 3\left(\Vert\mu_n-\mu\Vert_\mathbb{B}^2+16k^2T_n+k^2\alpha_j\Vert\nu_j\Vert_\mathbb{B}^2\right)
\le C\,(1+T_n),
\]
for a constant $C$ not depending on $n$ or $\omega$. Since $\mathbb{E}[\Vert\xi\Vert_\mathbb{B}^2]<\infty$ by Lemma~\ref{lem:boundedness_fourth_moment}, $T_n$ is the sample mean of i.i.d.\@ non-negative integrable random variables; therefore $T_n\to\mathbb{E}[\Vert\xi-\mu\Vert_\mathbb{B}^2]$ a.s.\@ and $\mathbb{E}[T_n]=\mathbb{E}[\Vert\xi-\mu\Vert_\mathbb{B}^2]$ for every $n$, so that $\{T_n\}_n$ converges in $L^1(\mathbb{P})$ (see e.g.\@ \citealp[Theorem~3, pag.~257]{shi2016}) and is uniformly integrable. Being dominated by a uniformly integrable family, $\{\Vert m_{j,n}^k - m_{j}^k \Vert_{\mathbb{B}}^2\}_n$ is uniformly integrable as well, and $\mathbb{E}[\Vert m_{j,n}^k - m_{j}^k \Vert_{\mathbb{B}}^2]\to0$; in particular, since $\Vert\cdot\Vert_\mathbb{H}\le\Vert\cdot\Vert_\mathbb{B}$, also $\mathbb{E}[\Vert m_{j,n}^k - m_{j}^k \Vert_{\mathbb{H}}^2]\to0$.
Finally, $M_{j,n}^k,M_j^k\in\mathcal{L}or$ (Lemma~\ref{lem:surjectivity}) take values in $[0,1]$, so that $\Vert M_{j,n}^k-M_j^k\Vert_\infty\le1$ deterministically for every $n$. Since $\Vert M_{j,n}^k-M_j^k\Vert_\infty\to0$ a.s.\@ (shown above), the Bounded Convergence Theorem gives
\begin{equation*}
	\mathbb{E} \left[ \left\Vert M_{j,n}^k - M_{j}^k \right\Vert_\infty^2 \right] \to 0 \quad \text{as } n \to \infty,
\end{equation*}
which concludes the proof.
\end{proof}

\begin{proof}[Proof of Proposition~\ref{pro:cons-est-MoV}.]
	Throughout the proof, $\widehat{\nu}_{j,n}$ is taken such that $\langle\widehat{\nu}_{j,n},\nu_j\rangle_\mathbb{H}\ge0$, the sign $\pm$ is the same in $\widehat{m}_{j,n}^k$ and in $m_j^k$, and $C$ denotes a universal constant whose value may change from line to line. The notation of the proof of Proposition~\ref{pro:cons-MoV-full-observed-LCs} is used; in particular, $T_n=n^{-1}\sum_{i=1}^n\Vert\xi_i-\mu\Vert_\mathbb{B}^2$ converges a.s.\@ and is uniformly integrable. Recall also that $\Vert\cdot\Vert_\mathbb{H}\le\Vert\cdot\Vert_\mathbb{B}$ and $\Vert\cdot\Vert_{L^1_{[0,1]}}\le\Vert\cdot\Vert_{L^2_{[0,1]}}$, and set
	\[
	e_i:=\widehat{\xi}_i-\xi_i,\qquad
	\zeta_i:=\xi_i-\mu_n,\qquad
	\widehat{\zeta}_i:=\widehat{\xi}_i-\widehat{\mu}_n,\qquad i=1,\ldots,n,
	\]
	so that $\widehat{\mu}_n-\mu_n=n^{-1}\sum_{i=1}^n e_i$ and $\widehat{\zeta}_i-\zeta_i=e_i-(\widehat{\mu}_n-\mu_n)$.
	
	As in Proposition~\ref{pro:cons-MoV-full-observed-LCs}, since $M_j^k=\psi^{-1}(m_j^k)$ and $\widehat{M}_{j,n}^k=\psi^{-1}(\widehat{m}_{j,n}^k)$, the continuity of $\psi^{-1}$ from $\mathbb{B}$ to $\left(\mathcal{L}or,\Vert\cdot\Vert_\infty\right)$ (see Proposition~\ref{pro:psi_bi-continuous_bijection}) implies that the a.s.\@ convergence of $\widehat{m}_{j,n}^k$ to $m_j^k$ in $\mathbb{B}$ is sufficient to guarantee that $\Vert\widehat{M}_{j,n}^k-M_j^k\Vert_\infty\to0$ a.s..
	To this end, consider
	\begin{equation}
		\Vert\widehat{m}_{j,n}^k-m_j^k\Vert_\mathbb{B}\le
		\underbrace{\Vert\widehat{\mu}_n-\mu_n\Vert_\mathbb{B}+\Vert\mu_n-\mu\Vert_\mathbb{B}}_{\text{(I)}}
		+k\underbrace{\left\vert\widehat{\alpha}_{j,n}^{1/2}-\alpha_j^{1/2}\right\vert\Vert\widehat{\nu}_{j,n}\Vert_\mathbb{B}}_{\text{(II)}}
		+k\,\alpha_j^{1/2}\underbrace{\Vert\widehat{\nu}_{j,n}-\nu_j\Vert_\mathbb{B}}_{\text{(III)}}.
		\label{eq:tri_ineq_est_MoV}
	\end{equation}
	The three terms are handled separately. Term~(I) is controlled directly by Assumption~\ref{hp:estimator_moment}\ref{ass:D_n} and by Proposition~\ref{pro:cons-MoV-full-observed-LCs}. Terms~(II) and~(III), instead, both depend on the eigenelements, and are therefore reached in three steps: the distance $\Vert\widehat{\Sigma}_{\xi,n}-\Sigma_\xi\Vert_\infty$ between the estimated and the theoretical covariance operators is bounded first; the usual perturbation bounds then give the convergence of $\widehat{\alpha}_{j,n}$ and of $\widehat{\nu}_{j,n}$ in $\mathbb{H}$, which settles the scalar factor in~(II); finally, the eigen-equation is used to upgrade the convergence of $\widehat{\nu}_{j,n}$ to the sup norm, which completes both~(II) and~(III). The $L^2(\mathbb{P})$ statement is proved at the end, by dominating~\eqref{eq:tri_ineq_est_MoV} through two uniformly integrable averages.
		
	Regarding the mean term (I), by Jensen's inequality,
	\begin{equation}\label{eq:est_mean}
		\Vert\widehat{\mu}_n-\mu_n\Vert_\mathbb{B}\le\frac1n\sum_{i=1}^n\Vert e_i\Vert_\mathbb{B}\le D_n^{1/2},
	\end{equation}
	which vanishes a.s.\@ by Assumption~\ref{hp:estimator_moment}\ref{ass:D_n}, while $\Vert\mu_n-\mu\Vert_\mathbb{B}\to0$ a.s.\@ by Proposition~\ref{pro:cons-MoV-full-observed-LCs}. Hence (I) vanishes a.s..
	
	For the covariance operator, let $a\otimes b$ be the operator $v\mapsto\langle b,v\rangle_\mathbb{H}\,a$, whose norm equals $\Vert a\Vert_\mathbb{H}\Vert b\Vert_\mathbb{H}$. Thus, one has $\Sigma_{\xi,n}=n^{-1}\sum_{i=1}^n\zeta_i\otimes\zeta_i$ and $\widehat{\Sigma}_{\xi,n}=n^{-1}\sum_{i=1}^n\widehat{\zeta}_i\otimes\widehat{\zeta}_i$, so that
	\begin{equation}\label{eq:otimes_decomposition}
		\widehat{\zeta}_i\otimes\widehat{\zeta}_i-\zeta_i\otimes\zeta_i=(\widehat{\zeta}_i-\zeta_i)\otimes\widehat{\zeta}_i+\zeta_i\otimes(\widehat{\zeta}_i-\zeta_i)
	\end{equation}
	and, consequently, $\Vert\widehat{\Sigma}_{\xi,n}-\Sigma_{\xi,n}\Vert_\infty\le R_n$, where
	\begin{equation}\label{eq:def_Rn}
		R_n:=\frac1n\sum_{i=1}^n\Vert\widehat{\zeta}_i-\zeta_i\Vert_\mathbb{B}\left(\Vert\widehat{\zeta}_i\Vert_\mathbb{B}+\Vert\zeta_i\Vert_\mathbb{B}\right).
	\end{equation}
	To see that $R_n$ vanishes, set $S_n:=n^{-1}\sum_{i=1}^n\Vert\zeta_i\Vert_\mathbb{B}^2$ and $\widehat{S}_n:=n^{-1}\sum_{i=1}^n\Vert\widehat{\zeta}_i\Vert_\mathbb{B}^2$, and note that $\widehat{\zeta}_i-\zeta_i=e_i-(\widehat{\mu}_n-\mu_n)$, so that $\Vert\widehat{\zeta}_i-\zeta_i\Vert_\mathbb{B}^2\le2\Vert e_i\Vert_\mathbb{B}^2+2\Vert\widehat{\mu}_n-\mu_n\Vert_\mathbb{B}^2$ by the triangle inequality and $(a+b)^2\le2a^2+2b^2$; averaging over $i$ and using~\eqref{eq:est_mean} gives
	\begin{equation}\label{eq:S_bounds}
		\frac1n\sum_{i=1}^n\Vert\widehat{\zeta}_i-\zeta_i\Vert_\mathbb{B}^2\le4D_n,
		\qquad
		S_n\le4T_n,
		\qquad
		\widehat{S}_n\le2S_n+8D_n\le C\,(T_n+D_n),
	\end{equation}
	the bound on $S_n$ following from the second inequality in~\eqref{eq:CS_eigen_bound}. By the Cauchy--Schwarz inequality applied to~\eqref{eq:def_Rn},
	\begin{equation}\label{eq:est_cross}
		R_n\le\left(\frac1n\sum_{i=1}^n\Vert\widehat{\zeta}_i-\zeta_i\Vert_\mathbb{B}^2\right)^{1/2}\left(\frac2n\sum_{i=1}^n\left(\Vert\widehat{\zeta}_i\Vert_\mathbb{B}^2+\Vert\zeta_i\Vert_\mathbb{B}^2\right)\right)^{1/2}
		\le C\,D_n^{1/2}\,(T_n+D_n)^{1/2},
	\end{equation}
	which vanishes a.s.\@ since $D_n\to0$ a.s.\@ and $T_n$ converges a.s.. Combined with $\Vert\Sigma_{\xi,n}-\Sigma_\xi\Vert_\infty\to0$ a.s.\@ (see again the proof of Proposition~\ref{pro:cons-MoV-full-observed-LCs}), this gives $\Vert\widehat{\Sigma}_{\xi,n}-\Sigma_\xi\Vert_\infty\to0$ a.s..
	
	Consider now the behaviour of eigenelements in $\mathbb{H}$. By the perturbation bounds for the eigenelements of compact self-adjoint operators \citep[Chapter~4]{bos00},
	\begin{align}
		&\left\vert\widehat{\alpha}_{j,n}-\alpha_j\right\vert\le\Vert\widehat{\Sigma}_{\xi,n}-\Sigma_\xi\Vert_\infty\to0,\qquad a.s.,
		\label{eq:est_eigenvalues}
		\\
		&\Vert\widehat{\nu}_{j,n}-\nu_j\Vert_\mathbb{H}\le\frac{2\sqrt2}{\delta_j}\Vert\widehat{\Sigma}_{\xi,n}-\Sigma_\xi\Vert_\infty\to0,\qquad a.s..
		\label{eq:est_eigenvectors}
	\end{align}
	In particular $\widehat{\alpha}_{j,n}\to\alpha_j>0$ a.s., so that $\widehat{\alpha}_{j,n}>0$ eventually and, since $\vert\widehat{\alpha}_{j,n}^{1/2}-\alpha_j^{1/2}\vert\le\vert\widehat{\alpha}_{j,n}-\alpha_j\vert\,\alpha_j^{-1/2}$, the scalar factor in (II) vanishes a.s..
	
	Regarding the behaviour of the eigenfunctions in $\mathbb{B}$, the sup-norm upgrade of~\eqref{eq:est_eigenvectors} also shows that $\Vert\widehat{\nu}_{j,n}\Vert_\mathbb{B}$ is eventually bounded, as required by~(II). Let $\widehat{c}_{2,1,n}$ and $\widehat{c}_n$ be defined as $c_{2,1,n}$ and $c_n$ in the proof of Proposition~\ref{pro:cons-MoV-full-observed-LCs}, with $\zeta_i$ replaced by $\widehat{\zeta}_i$. The scalar decomposition corresponding to~\eqref{eq:otimes_decomposition}, namely
	\[
	\widehat{\zeta}_i^{(2)}(x_1)\widehat{\zeta}_i^{(2)}(x_2)-\zeta_i^{(2)}(x_1)\zeta_i^{(2)}(x_2)
	=\left(\widehat{\zeta}_i^{(2)}-\zeta_i^{(2)}\right)(x_1)\,\widehat{\zeta}_i^{(2)}(x_2)+\zeta_i^{(2)}(x_1)\left(\widehat{\zeta}_i^{(2)}-\zeta_i^{(2)}\right)(x_2),
	\]
	and its analogue for the cross-covariance give $\Vert\widehat{c}_n-c_n\Vert_\infty\le R_n$ and $\Vert\widehat{c}_{2,1,n}-c_{2,1,n}\Vert_\infty\le R_n$. Hence, by~\eqref{eq:est_cross} and the uniform convergence of $c_n$ and $c_{2,1,n}$ established in Proposition~\ref{pro:cons-MoV-full-observed-LCs},
	\[
	\Vert\widehat{c}_n-c\Vert_\infty\to0,
	\qquad
	\Vert\widehat{c}_{2,1,n}-\Sigma_{\xi^{(2)},\xi^{(1)}}\Vert_\infty\to0,
	\qquad a.s.,
	\]
	and, in particular, $\Vert\widehat{c}_n\Vert_\infty$ and $\Vert\widehat{c}_{2,1,n}\Vert_\infty$ are eventually bounded a.s..
	Recalling that 
	\[
	\widehat{\Sigma}_{\xi,n}v=n^{-1}\sum_{i=1}^n\langle\widehat{\zeta}_i,v\rangle_\mathbb{H}\,\widehat{\zeta}_i
	\]
	and exchanging the finite sum with the integral in $\langle\widehat{\zeta}_i,v\rangle_\mathbb{H}$, the second component of $\widehat{\Sigma}_{\xi,n}v$ at $x\in[0,1]$ equals $v^{(1)}\widehat{c}_{2,1,n}(x)+\int_0^1\widehat{c}_n(x,y)v^{(2)}(y)\,dy$, which is the empirical analogue of the second block row of~\eqref{eq:block_theoretical_cov}. Evaluating it at $v=\widehat{\nu}_{j,n}$ and $v=\nu_j$ respectively, the functional components of the eigen-equations $\widehat{\Sigma}_{\xi,n}\widehat{\nu}_{j,n}=\widehat{\alpha}_{j,n}\widehat{\nu}_{j,n}$ and $\Sigma_\xi\nu_j=\alpha_j\nu_j$ read, for $x\in[0,1]$,
	\[
	\widehat{\alpha}_{j,n}\widehat{\nu}_{j,n}^{(2)}(x)=\widehat{\nu}_{j,n}^{(1)}\,\widehat{c}_{2,1,n}(x)+\int_0^1\widehat{c}_n(x,y)\,\widehat{\nu}_{j,n}^{(2)}(y)\,dy,
	\qquad
	\alpha_j\nu_j^{(2)}(x)=\nu_j^{(1)}\,\Sigma_{\xi^{(2)},\xi^{(1)}}(x)+\int_0^1 c(x,y)\,\nu_j^{(2)}(y)\,dy,
	\]
	which, as in Proposition~\ref{pro:eigenfunction_continuous}, also show that $\widehat{\nu}_{j,n}\in\mathbb{B}$ whenever $\widehat{\alpha}_{j,n}>0$. Subtracting the two equations and adding and subtracting intermediate terms,
	\begin{align*}
		\left\Vert\widehat{\alpha}_{j,n}\widehat{\nu}_{j,n}^{(2)}-\alpha_j\nu_j^{(2)}\right\Vert_\infty
		\le{}&\left\vert\widehat{\nu}_{j,n}^{(1)}-\nu_j^{(1)}\right\vert\,\Vert\widehat{c}_{2,1,n}\Vert_\infty
		+\left\vert\nu_j^{(1)}\right\vert\,\Vert\widehat{c}_{2,1,n}-\Sigma_{\xi^{(2)},\xi^{(1)}}\Vert_\infty\\
		&+\Vert\widehat{c}_n\Vert_\infty\,\Vert\widehat{\nu}_{j,n}^{(2)}-\nu_j^{(2)}\Vert_{L^2_{[0,1]}}
		+\Vert\widehat{c}_n-c\Vert_\infty\,\Vert\nu_j^{(2)}\Vert_{L^2_{[0,1]}},
	\end{align*}
	where the four terms on the right-hand side all vanish a.s.. Indeed, in the first and third term $\vert\widehat{\nu}_{j,n}^{(1)}-\nu_j^{(1)}\vert$ and $\Vert\widehat{\nu}_{j,n}^{(2)}-\nu_j^{(2)}\Vert_{L^2_{[0,1]}}$ are both bounded by $\Vert\widehat{\nu}_{j,n}-\nu_j\Vert_\mathbb{H}$, which vanishes a.s.\@ by~\eqref{eq:est_eigenvectors}, while the factors $\Vert\widehat{c}_{2,1,n}\Vert_\infty$ and $\Vert\widehat{c}_n\Vert_\infty$ are eventually bounded a.s.\@ as noted above; in the second and fourth term the factors $\vert\nu_j^{(1)}\vert$ and $\Vert\nu_j^{(2)}\Vert_{L^2_{[0,1]}}$ are bounded by $\Vert\nu_j\Vert_\mathbb{H}=1$, while $\Vert\widehat{c}_{2,1,n}-\Sigma_{\xi^{(2)},\xi^{(1)}}\Vert_\infty$ and $\Vert\widehat{c}_n-c\Vert_\infty$ vanish a.s.\@ by the convergence of the empirical covariance functions established above. 
	Since, for $\widehat{\alpha}_{j,n}>0$,
	\[
	\widehat{\nu}_{j,n}^{(2)}-\nu_j^{(2)}=\frac{1}{\widehat{\alpha}_{j,n}}\left(\widehat{\alpha}_{j,n}\widehat{\nu}_{j,n}^{(2)}-\alpha_j\nu_j^{(2)}\right)+\frac{\alpha_j-\widehat{\alpha}_{j,n}}{\widehat{\alpha}_{j,n}}\,\nu_j^{(2)},
	\]
	the triangle inequality gives
	\[
	\Vert\widehat{\nu}_{j,n}^{(2)}-\nu_j^{(2)}\Vert_\infty\le\frac{1}{\widehat{\alpha}_{j,n}}\left\Vert\widehat{\alpha}_{j,n}\widehat{\nu}_{j,n}^{(2)}-\alpha_j\nu_j^{(2)}\right\Vert_\infty+\left\vert\frac{\alpha_j}{\widehat{\alpha}_{j,n}}-1\right\vert\Vert\nu_j^{(2)}\Vert_\infty,
	\]
	where $\widehat{\alpha}_{j,n}\ge\alpha_j/2$ eventually a.s.\@ and $\alpha_j/\widehat{\alpha}_{j,n}\to1$ a.s.\@ by~\eqref{eq:est_eigenvalues}, while $\Vert\nu_j^{(2)}\Vert_\infty<\infty$ by Proposition~\ref{pro:eigenfunction_continuous}. It follows that $\Vert\widehat{\nu}_{j,n}-\nu_j\Vert_\mathbb{B}\to0$ a.s..
	All three terms in~\eqref{eq:tri_ineq_est_MoV} therefore vanish, so that $\Vert\widehat{m}_{j,n}^k-m_j^k\Vert_\mathbb{B}\to0$ a.s.\@ and, by the continuity of $\psi^{-1}$ recalled at the beginning of the proof, $\Vert\widehat{M}_{j,n}^k-M_j^k\Vert_\infty\to0$ a.s..
	
	Finally, regarding the convergence in $L^2(\mathbb{P})$, note that, since $\Vert\widehat{m}_{j,n}^k-m_j^k\Vert_\mathbb{B}\to0$ a.s.\@ and $\{\Vert\widehat{m}_{j,n}^k-m_j^k\Vert_\mathbb{B}^2\}_n$ is non-negative, it is enough to prove that this family is uniformly integrable (see \citealp[Theorem 4, p.~188]{shi2016}). Hence, the strategy is to dominate it by the two averages
	\[
	T_n=\frac1n\sum_{i=1}^n\Vert\xi_i-\mu\Vert_\mathbb{B}^2,
	\qquad
	D_n=\frac1n\sum_{i=1}^n\Vert\widehat{\xi}_i-\xi_i\Vert_\mathbb{B}^2,
	\]
	the former accounting for the sampling of the \ac{LCs}, the latter for their estimation, and to show that each of them is uniformly integrable.
	The only term in~\eqref{eq:tri_ineq_est_MoV} that is not already dominated by $T_n$ or $D_n$ is the factor $\widehat{\alpha}_{j,n}^{1/2}\Vert\widehat{\nu}_{j,n}\Vert_\mathbb{B}$ appearing in (II). It is bounded exactly as in~\eqref{eq:CS_eigen_bound}, where the Cauchy--Schwarz inequality applied to the componentwise eigen-equations makes the factor $\widehat{\alpha}_{j,n}^{1/2}$ cancel out: replacing $\Sigma_{\xi,n}$, $\alpha_{j,n}$, $\nu_{j,n}$ and $\xi_i-\mu_n$ by $\widehat{\Sigma}_{\xi,n}$, $\widehat{\alpha}_{j,n}$, $\widehat{\nu}_{j,n}$ and $\widehat{\zeta}_i$ gives
	\[
	\widehat{\alpha}_{j,n}^{1/2}\,\Vert\widehat{\nu}_{j,n}\Vert_\mathbb{B}\le2\,\widehat{S}_n^{1/2}\le C\,(T_n+D_n)^{1/2},
	\]
	the last inequality being~\eqref{eq:S_bounds}. Since also $\Vert\widehat{\mu}_n-\mu_n\Vert_\mathbb{B}^2\le D_n$ by~\eqref{eq:est_mean}, $\Vert\mu_n-\mu\Vert_\mathbb{B}^2\le T_n$ by Jensen's inequality, and $\alpha_j\Vert\nu_j\Vert_\mathbb{B}^2$ is a deterministic constant, squaring~\eqref{eq:tri_ineq_est_MoV} yields
	\begin{equation}\label{eq:est_domination}
		\Vert\widehat{m}_{j,n}^k-m_j^k\Vert_\mathbb{B}^2\le C\,(1+T_n+D_n),
	\end{equation}
	where $C$ depends on $k$, $\alpha_j$ and $\nu_j$, but neither on $n$ nor on $\omega$.
	The family $\{T_n\}_n$ is uniformly integrable, as shown in the proof of Proposition~\ref{pro:cons-MoV-full-observed-LCs}, because the $\xi_i$ are i.i.d.\@ and $\mathbb{E}[\Vert\xi\Vert_\mathbb{B}^2]<\infty$. The same argument is not available for $\{D_n\}_n$, since the $\widehat{\xi}_i$ form a triangular array and are neither independent nor identically distributed; uniform integrability is obtained instead from a uniform bound on the second moment. Indeed, by Jensen's inequality and $\Vert\widehat{\xi}_i-\xi_i\Vert_\mathbb{B}^4\le8(\Vert\widehat{\xi}_i\Vert_\mathbb{B}^4+\Vert\xi_i\Vert_\mathbb{B}^4)$,
	\[
	\sup_n\mathbb{E}\left[D_n^2\right]\le\sup_n\frac1n\sum_{i=1}^n\mathbb{E}\left[\Vert\widehat{\xi}_i-\xi_i\Vert_\mathbb{B}^4\right]\le8\left(\sup_{n\ge1}\max_{1\le i\le n}\mathbb{E}\left[\Vert\widehat{\xi}_i\Vert_\mathbb{B}^4\right]+\mathbb{E}\left[\Vert\xi\Vert_\mathbb{B}^4\right]\right)<\infty,
	\]
	which is finite by Assumption~\ref{hp:estimator_moment}\ref{ass:moment_4_hat} and Lemma~\ref{lem:boundedness_fourth_moment}; a family bounded in $L^2(\mathbb{P})$ is uniformly integrable.
	Since the sum of two uniformly integrable families is uniformly integrable, and so is any family dominated by one, \eqref{eq:est_domination} gives the uniform integrability of $\{\Vert\widehat{m}_{j,n}^k-m_j^k\Vert_\mathbb{B}^2\}_n$ and hence $\mathbb{E}[\Vert\widehat{m}_{j,n}^k-m_j^k\Vert_\mathbb{B}^2]\to0$.
	Finally, $\widehat{M}_{j,n}^k,M_j^k\in\mathcal{L}or$ (Lemma~\ref{lem:surjectivity}) take values in $[0,1]$, so that $\Vert\widehat{M}_{j,n}^k-M_j^k\Vert_\infty\le1$ deterministically for every $n$. Since $\Vert\widehat{M}_{j,n}^k-M_j^k\Vert_\infty\to0$ a.s., the Bounded Convergence Theorem gives
	\begin{equation*}
		\mathbb{E}\left[\left\Vert \widehat{M}_{j,n}^k-M_{j}^k \right\Vert_\infty^2\right] \to 0 \quad\text{as } n\to\infty,
	\end{equation*}
	which concludes the proof.
\end{proof}

\bibliography{Lorenz_biblio}

\end{document}